\documentclass[11pt,a4paper]{article}

\usepackage[a4paper,margin=2.15cm]{geometry}
\usepackage{setspace}
\usepackage{amsmath,amssymb,amsthm,mathtools}
\numberwithin{equation}{section}
\usepackage{bm}
\usepackage{tensor}
\usepackage{slashed}
\usepackage{cancel}

\usepackage{booktabs}
\usepackage{array}
\usepackage[dvipsnames]{xcolor}
\usepackage{graphicx}
\usepackage{caption}
\usepackage{subcaption}
\usepackage{tabularx}
\usepackage{booktabs}

\usepackage{comment}

\usepackage[colorlinks=true,linkcolor=blue,citecolor=blue,urlcolor=blue]{hyperref}

\usepackage[numbers,sort&compress]{natbib}

\newcommand{\M}{\mathcal{M}}
\newcommand{\B}{\mathcal{B}}
\newcommand{\T}{\mathcal{T}}
\newcommand{\K}{\mathcal{K}}
\newcommand{\nn}{\nonumber}
\newcommand{\defeq}{\coloneqq}

\newcolumntype{L}[1]{>{\raggedright\arraybackslash}p{#1}}
\newcolumntype{C}[1]{>{\centering\arraybackslash}p{#1}}
\newcolumntype{Y}{>{\raggedright\arraybackslash}X}

\newtheorem{theorem}{Theorem}
\newtheorem{corollary}{Corollary}
\newtheorem{prop}{Proposition}
\newtheorem{lemma}{Lemma}

\theoremstyle{definition}

\theoremstyle{remark}
\newtheorem{remark}{Remark}

\title{\bf Thermomechanics of Horndeski Scalar Hair\\ with Spacelike Gradients}

\author{Marcello Miranda$^{a,b}$,
$\ $
Marco Sebastianutti$^{c,d}$,
$\ $
Andrea Giusti$^{c,d,e}$\thanks{E-mail: andrea.giusti9@unibo.it},
\\
\\
$^a${\em Scuola Superiore Meridionale, Via Mezzocannone 4, I-80134 Napoli, Italy}
\\
\\
$^b${\em I.N.F.N., Sezione di di Napoli, Via Cinthia 9, I-80126 Napoli, Italy}
\\
\\
$^c${\em DIFA, University of Bologna, I-40126 Bologna, Italy}
\\
\\
$^d${\em I.N.F.N., Sezione di Bologna, I.S.~FLAG, I-40127 Bologna, Italy}
\\
\\
$^e${\em Alma Mater Research Center on Applied Mathematics $({\cal AM}^2)$,}\\ 
{\em University of Bologna, I-40123 Bologna, Italy}
}

\date{ }

\begin{document}

\maketitle

\begin{abstract}
\noindent We develop a thermomechanical framework for viable Horndeski scalar hair with spacelike scalar-field gradients.
Adapting the $(1{+}1{+}2)$ formalism to the scalar configuration, we apply the effective imperfect-fluid decomposition to the model and derive general results governing its thermomechanical properties.
Specifically, we infer the constitutive relations for the effective medium, in both observer-dependent and observer-independent approaches, thereby identifying the conditions under which the scalar hair admits thermomechanical interpretations.
The observer-dependent approach leads to an analogy with anisotropic media characterized by a preferred spatial direction, whereas the observer-independent view shows similarities with layered media. The general results are then specialized to static, spherically symmetric configurations with scalar hair depending solely on the areal radius.
Interestingly, this scenario allows us to connect the thermomechanical framework with the Tolman--Ehrenfest criterion for thermal equilibrium.
\end{abstract}

\newpage

\section{Introduction}
\label{sec:intro}

Scalar-tensor theories are among the simplest extensions of General Relativity obtained by introducing an additional scalar gravitational degree of freedom. The earliest fully developed proposal in this family is the celebrated Brans--Dicke theory~\cite{Brans:1961sx}, while Horndeski gravity provides the most general four-dimensional scalar-tensor theory with a single scalar field and second-order Euler--Lagrange equations under the standard assumptions of locality and diffeomorphism invariance~\cite{Horndeski:1974wa}. Its modern formulation in terms of the functions $G_i(\phi,X)$ has made it a standard framework for cosmology, compact objects, and tests of gravity~\cite{Deffayet:2011gz, Kobayashi:2011nu, Kobayashi:2019hrl}. Following the multimessenger constraint from GW170817~\cite{LIGOScientific:2017vwq, LIGOScientific:2017zic, Creminelli:2017sry, Ezquiaga:2017ekz}, we restrict attention to viable Horndeski gravity, selected by requiring luminal propagation of gravitational waves and described by a $G_2(\phi,X)$ $k$-essence term, a $-G_3(\phi,X)\Box\phi$ cubic Horndeski term, and a conformal nonminimal coupling to the Ricci scalar $G_4(\phi)R$.

A useful way of analyzing scalar-tensor dynamics is to move all non-Einstein terms to the right-hand side of the metric field equations and interpret them as an effective stress-energy tensor. Effective-fluid descriptions have been developed for minimally coupled scalar fields and kinetic-gravity-braiding theories~\cite{Pujolas:2011he, Gergely:2020ncx}. In nonminimally coupled scalar-tensor gravity, the effective source has also been fruitfully compared with a relativistic imperfect fluid~\cite{Faraoni:2018qdr, Faraoni:2021lfc, Faraoni:2021jri}. More recently, the constitutive structure of viable Horndeski gravity has been related to first-order relativistic thermodynamics~\cite{Giusti:2021sku, Miranda:2022wkz, Banerjee:2023ayn, Miranda:2024dhw, Gallerani:2024gdy, Pereira:2025dmk, Giusti:2026ymb}.

The {\em Thermodynamics of Scalar-Tensor Gravity}~\cite{Faraoni:2021lfc, Faraoni:2021jri} (see also~\cite{Giardino:2023ygc}), in its current formulation, relies on the assumption that the scalar-field gradient is timelike. Its normalized gradient then selects a natural four-velocity, and the scalar-comoving $(1{+}3)$ decomposition can be compared directly with the variables of a relativistic imperfect fluid. This has been shown to allow a formal analogy between the constitutive laws of the resulting effective imperfect fluid and those of Eckart's thermodynamics~\cite{Faraoni:2021lfc, Faraoni:2021jri}. Notably, this analogy has been extended to map onto the causal M{\"u}ller--Israel--Stewart model~\cite{Banerjee:2026ulx}. 

The scope of this work is to extend the thermomechanical analysis, based on the imperfect-fluid representation of scalar-tensor theories, to the case of scalar fields with spacelike gradients. Accordingly, we assume that
\begin{equation}
\nabla_a\phi\nabla^a\phi>0\, .
\end{equation}
The normalized scalar gradient defines a preferred spatial direction, which we shall refer to as the \emph{scalar director}, $s^a\propto\nabla^a\phi$, but does not select a unique timelike congruence. The orthogonal complement of the scalar gradient,
\begin{equation}
{\rm Ker} \,(\mathrm d\phi) \equiv (\nabla\phi)^\perp\subset TM\,,
\end{equation}
is, throughout the spacelike-gradient region, a rank-3 Lorentzian subbundle. It contains a family of admissible timelike observers related locally by proper orthochronous ${\rm SO}^\uparrow(1,2)$ transformations. As a consequence, the same effective stress-energy tensor admits different anisotropic-fluid decompositions depending on the choice of the unit timelike vector field tangent to the scalar leaves $\Sigma_c \defeq \phi^{-1} (c) \subset M$, with $c\in \mathbb{R}$. 

Spacelike scalar gradients have previously been considered for minimally coupled scalar fields, where they yield a type-I anisotropic-fluid description~\cite{Gergely:2020ncx}. Extending this picture to viable Horndeski gravity requires accounting for the additional nonminimal-coupling and kinetic-braiding contributions.

This spacelike-gradient setting arises naturally in static exterior regions of spherically symmetric configurations with a symmetry-inheriting radial scalar profile. It is therefore relevant to static branches of scalar hair in black-hole and compact-object spacetimes. Explicit examples and broader studies of hairy configurations have been presented in generalized scalar-tensor and Horndeski gravity~\cite{Sotiriou:2013qea, Babichev:2016rlq}, while systematic analyses have clarified the assumptions underlying hair and no-hair results~\cite{Faraoni:2021nhi, Capuano:2023yyh}. Static scalar-haired solutions in nonminimally coupled $k$-essence provide a further realization of the same spacelike-gradient setting~\cite{Nagy:2021bha}. Our purpose is not to construct or test a new compact-object solution, but to determine what effective continuum structure is supported by the viable Horndeski stress-energy tensor and which parts of it are independent of the chosen fluid observer.

The natural covariant framework for this problem is a $(1{+}1{+}2)$ split adapted to the unit spatial vector $s$~\cite{vanElst:1995eg, Clarkson:2007yp}. In addition to the usual $(1{+}3)$ splitting induced by the unit timelike vector $u$, the vector $s$ singles out a preferred direction within the three-dimensional spatial subspace orthogonal to $u$, further decomposing it into a longitudinal direction along $s$ and a two-dimensional transverse subspace, the so-called ``sheet'', orthogonal to both $u$ and $s$. In scalar-tensor gravity, this formalism has been applied to static and spherically symmetric solutions, Birkhoff-type results, asymptotic flatness, and quasi-local horizons~\cite{Carloni:2013iip}. 
We employ this framework to identify the complete set of effective energy, flux, and stress channels generated by a spacelike scalar gradient. 

Our first result is the full $(1{+}1{+}2)$ decomposition of the viable Horndeski effective scalar stress-energy tensor relative to an arbitrary unit timelike vector field orthogonal to the scalar gradient. We derive the effective energy density, longitudinal and transverse pressures, longitudinal and sheet energy fluxes, and the vector and sheet trace-free anisotropic stresses. 

The main finding is that these frame-dependent variables recombine into an observer-independent representation intrinsic to the timelike hypersurfaces $\Sigma_c$, with constant $c\in \mathbb{R}$. In this scalar-layer formulation, the effective source separates into a normal stress, a leaf-tangential traction, and a stress-energy tensor intrinsic to each scalar leaf. The intrinsic leaf stress contains a term proportional to the trace-adjusted extrinsic curvature of the foliation, while the leaf-tangential traction is locally exact on each leaf and can be written as the leaf gradient of a scalar potential. This representation isolates the geometrical and constitutive content fixed by the scalar field itself from the quantities introduced by a particular observer-dependent fluid decomposition.

The observer-independent form also allows for an invariant analysis of preferred scalar-adapted frames. The causal character of the leaf-tangential traction determines whether its longitudinal-flux projection or its mixed normal--sheet stress projection can be eliminated by changing the timelike observer. Moreover, the intrinsic leaf stress and the extrinsic curvature share their principal directions whenever the nonminimal-coupling response is nonvanishing. This yields geometric criteria for the existence of what we dub principal-curvature observers, leaf-Landau observers, and Landau observers of the complete effective scalar stress-energy tensor.

The paper is organized as follows. Sections~\ref{sec:kinematics} and~\ref{sec:decomposition} introduce the $(1{+}1{+}2)$ kinematics and stress-energy decomposition, while Sec.~\ref{sec:scalar} specializes the construction to a spacelike scalar gradient. Section~\ref{sec:horndeski-effective-stress} derives the viable Horndeski effective stress-energy tensor and its characteristic relations. Sections~\ref{sec:observer-dependent} and~\ref{sec:scalar-layer-representation} develop the observer-dependent and observer-independent representations, respectively. The static spherically symmetric sector is then considered, while the detailed classification of scalar-adapted frames is collected in the appendix. Throughout the paper, we use geometrized units $G=c=1$ and the standard general-relativistic notation of~\cite{Wald:1984rg} and~\cite{Chrusciel:2019lxd}. Unless otherwise stated, Greek indices denote the components in the specific coordinate chart, while Latin indices denote the use of the abstract index convention.

\medskip

\noindent {\bf Note added:} while completing this work, a closely related study~\cite{Gergely:2026drx} appeared on arXiv. Reference~\cite{Gergely:2026drx} focuses on kinetic gravity braiding, a subclass of the viable Horndeski model minimally coupled to gravity, and employs a $(2{+}1{+}1)$ decomposition to investigate the fluid representation of the theory for timelike, spacelike, and null scalar-field gradients. We therefore wish to recognize priority to the author of Ref.~\cite{Gergely:2026drx} for this application. Our analysis nevertheless differs in scope, as we consider the more general viable Horndeski theory, including the nonminimal coupling $G_4(\phi)R$, and focus specifically on the thermomechanical interpretation of the resulting effective stress-energy tensor, with the aim of extending the framework originally formulated in Refs.~\cite{Faraoni:2021lfc, Faraoni:2021jri} to spacelike scalar-field gradients. Furthermore, we emphasize the non-uniqueness of the timelike observer $u\in{\rm Ker}({\rm d}\phi)$ and develop an observer-independent interpretation of the effective medium that we dub the scalar-layer representation. The two analyses are therefore complementary, and we refer the reader to Ref.~\cite{Gergely:2026drx} for the treatment of minimally coupled kinetic gravity braiding.

\section{Kinematics of the \texorpdfstring{$(1{+}1{+}2)$}{(1+1+2)} split} \label{sec:kinematics}

We begin by specifying the basic assumptions that underpin the geometrical construction discussed in this section.

Let $M$ be a smooth, four-dimensional, time-orientable manifold equipped with a Lorentzian metric $g_{ab}$ of signature $(-,+,+,+)$. To avoid scenarios with distributional curvature, we require the metric to be at least $g_{ab} \in C^2(M)$ so that the Levi-Civita connection $\nabla_a$ and the Riemann tensor are globally well-defined and continuous. Additionally, this guarantees that all $(1{+}1{+}2)$ kinematic variables are pointwise well-defined and continuous. The local four-dimensional volume form is denoted by $\eta_{abcd}$. Note that, unless otherwise stated, all constructions below are understood locally. 

On a local open subset $U \subseteq M$ we denote by $u\in \Gamma(TU)$ a unit timelike vector field and by $s\in \Gamma(TU)$ a unit spacelike vector field, which satisfy
\begin{equation}
  u^a u_a=-1\,, 
  \qquad 
  s^a s_a=1\,,
  \qquad
  u^a s_a=0 \, .
\end{equation}
Such fields always exist locally on a Lorentzian manifold. A local timelike direction can be chosen, and a spacelike vector orthogonal to it can then be obtained by a local Gram–Schmidt orthogonalization. However, when extending these fields globally, one may encounter topological obstructions.

The tensor
\begin{equation}
h_{ab}\defeq g_{ab}+u_a u_b
\end{equation}
is the orthogonal projector onto the rank-3 subbundle
\begin{equation}
u^\perp \defeq  \bigcup _{p \in U} \{v\in T_p U \mid g_{ab} (p) \, u^a (p) \, v^b=0\} \, ,
\end{equation}
i.e., the local rest space of the observers with four-velocity $u$.\footnote{Note that $u^\perp$, the rank-3 subbundle of $TU$ (often referred to as a rank-3 distribution on $U$, see~\cite{Lee:2012ism}), need not be integrable, i.e., it does not necessarily admit a 3-dimensional submanifold $\Sigma \subseteq U$ such that $T\Sigma = u^\perp$.} 
Furthermore, the tensor
\begin{equation}
\gamma_{ab}\defeq h_{ab}-s_a s_b = g_{ab}+u_a u_b-s_a s_b \, ,
\end{equation}
defines the orthogonal projector onto the rank-2 subbundle
\begin{equation}
\mathcal{S}\defeq  (\mathrm{Span}(u,s))^\perp \, ,
\end{equation}
that we dub the sheet orthogonal to $\mathrm{Span}(u,s)$.

The induced area two-form on this two-dimensional subspace is
\begin{equation}
\epsilon_{ab}\defeq\eta_{abcd}u^c s^d  \, ,
\end{equation}
and satisfies
\begin{equation}
\epsilon_{ab}u^b=0 \, ,\qquad
\epsilon_{ab}s^b=0 \, ,
\end{equation}
and is the Levi--Civita tensor associated with the metric induced by $\gamma_{ab}$ on the sheet $\mathcal{S}$. 

It is also worth recalling some elementary identities for later convenience.
\begin{lemma}
The following local identities hold:
\begin{align}
  h_a{}^b u_b &=0,
  &
  h_a{}^b s_b &=s_a,
  \\
  \gamma_a{}^b u_b &=0,
  &
  \gamma_a{}^b s_b &=0,
  \\
  h_a{}^c h_c{}^b &=h_a{}^b,
  &
  \gamma_a{}^c\gamma_c{}^b &=\gamma_a{}^b,
  \\
  h^a{}_a&=3,
  &
  \gamma^a{}_a&=2.
\end{align}
\end{lemma}

We can now continue our presentation of the $(1{+}1{+}2)$ split by introducing the kinematic quantities associated with the timelike vector field $u$. Specifically, we define the acceleration, expansion, shear, and vorticity of $u$, respectively, as
\begin{align}
  \dot{u}_a &\defeq u^b\nabla_b u_a,
  \\
  \Theta &\defeq \nabla_a u^a,
  \\
  \sigma_{ab} &\defeq\left(h_{(a}{}^{c}h_{b)}{}^{d}-\frac{1}{3}h_{ab}h^{cd}\right)\nabla_c u_d\,,
  \\
  \omega_{ab} &\defeq D_{[b}u_{a]} = h_{[b}{}^{c}h_{a]}{}^d\nabla_c u_d,
\end{align}
where the brackets $(\,)$ and $[\,]$ denote normalized symmetrization and antisymmetrization, respectively, and $D$ is the fully projected spatial derivative onto $u^\perp$, i.e., 
$$ 
D_e T^{a_1\ldots a_p}{}_{b_1\ldots b_q}\defeq h_e{}^f h^{a_1}{}_{c_1}\cdots h^{a_p}{}_{c_p} h_{b_1}{}^{d_1}\cdots h_{b_q}{}^{d_q} \nabla_f T^{c_1\ldots c_p}{}_{d_1\ldots d_q} \, ,
$$ 
in particular, $D_a h_{bc}=0$.

One can then easily check the following:\footnote{Recalling the standard musical isomorphisms: given $v \in \Gamma (TU)$ and $\xi \in \Gamma (T^\ast U)$, we define $v^\flat \defeq  g(v, \cdot) \in \Gamma (T^\ast U)$ (i.e., $(v^\flat) _a = v_a = g_{ab} \, v^b$) and $\xi^\sharp \in \Gamma (T U)$ such that $g(\xi^\sharp, \cdot) = \xi (\cdot)$ (i.e., $(\xi^\sharp) ^a = \xi^a = g^{ab} \, \xi_b$). See e.g.~\cite{Chrusciel:2019lxd}.}
\begin{lemma}
The covariant derivative 
$\nabla u^\flat \in \Gamma (T^\ast U \otimes T^\ast U)$ decomposes as
\begin{equation}
  \nabla_b u_a
  =
  -u_b\dot{u}_a
  +\frac13\Theta h_{ab}
  +\sigma_{ab}
  +\omega_{ab}.
  \label{eq:du-decomposition}
\end{equation}    
\end{lemma}

For the spacelike direction $s$, we can introduce the following quantities:
\begin{align}
\dot{s}_a &\defeq u^b\nabla_b s_a\,, \label{eq:s-dot}\\[3pt]
\widehat{s}_a &\defeq s^bD_b s_a\,, \label{eq:s-hat}\\[3pt]
\widetilde{\Theta} &\defeq D_a s^a = h^{ab}D_a s_b = \gamma^{ab}D_a s_b\,, \label{eq:sheet-expansion}\\[3pt]
\widetilde{\omega} &\defeq \frac12\epsilon^{ab}\gamma_b{}^c\nabla_c s_a
= \frac12\epsilon^{ab}D_b s_a\,, \label{eq:sheet-twist} \\[3pt]
\widetilde{\sigma}_{ab} &\defeq 
\left( \gamma_{(a}{}^c\gamma_{b)}{}^d -\frac12\gamma_{ab}\gamma^{cd} \right) \nabla_c s_d \nonumber\\[3pt]
&\,= \left( \gamma_{(a}{}^c\gamma_{b)}{}^d -\frac12\gamma_{ab}\gamma^{cd} \right) D_c s_d\,, \label{eq:sheet-shear} \\[3pt]
\chi_a &\defeq h_a{}^c s^d\nabla_c u_d = \frac13\Theta s_a +\sigma_{ab}s^b - \omega_{ab}s^b\,. \label{eq:chi-definition}
\end{align}
Throughout the text, the \textit{overdot} and \textit{hat} symbols are adopted to indicate the directional derivatives $\nabla _u \equiv u^a\nabla_a$ and $D_s \equiv s^a D_a$, respectively.

Then one can easily check the following:
\begin{lemma}
The spatial derivative $D s^\flat \in \Gamma ((u^\perp)^\ast \otimes (u^\perp)^\ast)$ decomposes as
\begin{equation}
  D_b s_a
  =
  s_b\widehat{s}_a
  +\frac12\widetilde{\Theta}\,\gamma_{ab}
  +\widetilde{\omega}\,\epsilon_{ab}
  +\widetilde{\sigma}_{ab}\, , \label{eq:director-gradient-decomposition}
\end{equation}
while the full covariant derivative $\nabla s^\flat \in \Gamma (T^\ast U \otimes T^\ast U)$ decomposes as
\begin{equation}
  \nabla_b s_a
  =
  -u_b\dot{s}_a
  +s_b\widehat{s}_a
  +\frac12\widetilde{\Theta}\gamma_{ab}
  +\widetilde{\omega}\,\epsilon_{ab}
  +\widetilde{\sigma}_{ab}+u_{a}\chi_{b} \, .
  \label{eq:ds-decomposition}
\end{equation}
\end{lemma}

Now, it is important to note that the derivatives of $s^\flat$ and $u^\flat$ along $u$ decompose respectively as
\begin{align}
\dot{s}_a = \mathcal A u_a+\alpha_a\,, \qquad \alpha_a \defeq \gamma_a{}^b\dot{s}_b\,, \label{eq:s-dot-decomposition}\\
\dot{u}_a=\mathcal{A}s_a+\beta_a\,, \qquad \beta_a\defeq\gamma_a{}^b\dot{u}_b\, .\label{eq:u-dot-decomposition}
\end{align}
Let us point out that $\mathcal{A}=0=\beta_a$ holds if and only if $u$ is geodesic ($\dot{u} \equiv 0$), whereas if $u$ is non-geodesic then $\dot{u}$ is necessarily spacelike. Furthermore, we note that $\alpha \in \mathcal{S}$ is a sheet vector and that the normalization $s_a s^a=1$ implies $s^a\dot{s}_a=0$, while differentiation of $u^a s_a=0$ along $u$ gives
\begin{equation}
\mathcal A = s^a\dot{u}_a = -u^a\dot{s}_a\,.
\end{equation}
It is also worth distinguishing the quantity $\widehat{s}_a=s^bD_b s_a$, which describes the {\em spatial bending} of the integral curves of $s$ within the spatial hypersurfaces orthogonal to $u$ (assuming that the rank-3 distribution $u^\perp$ is integrable). It should be distinguished from the spacetime curvature vector $s^b\nabla_b s_a$, which describes their curvature in the full spacetime. Then, contracting the full decomposition of $\nabla_b s_a$ with $s^b$ gives
\begin{equation}
s^b\nabla_b s_a = \widehat{s}_a + u_a s^b\chi_b\,. \label{eq:sds-decomposition} 
\end{equation}
The vector $\chi \in \Gamma (u^\perp)$ could be further decomposed, but it is not useful for our purposes.

\section{Imperfect fluid decomposition: \texorpdfstring{$(1{+}3)$}{(1+3)} and \texorpdfstring{$(1{+}1{+}2)$}{(1+1+2)} formulations} \label{sec:decomposition}

Let ${\rm T} \in \Gamma ({\rm Sym}^2 (T^\ast U))$ be a symmetric rank-$(0,2)$ tensor, meant to describe an effective stress-energy tensor. Then its $(1{+}3)$ decomposition with respect to $u$ reads
$$
{\rm T} = \rho \, u^\flat \otimes u^\flat
  +
  p \, {\rm h}
  +
  2 \, {\rm Sym} (q \otimes u^\flat)
  +
  \pi\,,
$$
or, equivalently, in abstract index notation
\begin{equation}
  T_{ab}
  =
  \rho \, u_a u_b
  +
  p \, h_{ab}
  +
  2q_{(a}u_{b)}
  +
  \pi_{ab}\,,\label{eq:T13-decomposition}
\end{equation}
where
\begin{align}
  \rho &\defeq T_{ab}u^a u^b,
  \\
  p &\defeq \frac13 h^{ab}T_{ab},
  \\
  q_a &\defeq -h_a{}^cT_{cd}u^d,
  \\
  \pi_{ab}
  &\defeq
  \left(h_a{}^c h_b{}^d
  -
  \frac13 h_{ab}h^{cd}\right)T_{cd}.
\end{align}
Here, $\rho$ denotes the effective energy density, $p$ is the isotropic pressure, $q_a$ is the heat-flux density 1-form, and $\pi_{ab}$ denotes the anisotropic stress tensor (which is symmetric and trace-free).

This decomposition holds for any symmetric rank-$(0,2)$ tensor, upon choosing the 4-velocity $u$ of a timelike observer (see~\cite{Faraoni:2023hwu}).

In the $(1{+}1{+}2)$ split adapted to $s$, we define the quantities
\begin{align}
  \mathcal Q &\defeq -T_{ab}u^a s^b\,,
  \\
  \mathcal Q_a &\defeq -T_{bc}u^b \gamma^c{}_{a}\,,
  \\
  p_{\parallel} &\defeq T_{ab}s^a s^b\,,
  \\
  p_{\perp} &\defeq \frac12\gamma^{ab}T_{ab}\,.
\end{align}
As a consequence, the isotropic pressure and heat flux split into the sum of their longitudinal and transverse parts, i.e.,
\begin{equation}
    p=\frac13\left(p_{\parallel}+2 p_{\perp}\right),
\end{equation}
\begin{equation}
\label{eq:dec-q}
  q_a = \mathcal Q s_a + \mathcal Q_a, \qquad \mathcal Q\defeq q_a s^a, \qquad \mathcal Q_a\defeq \gamma_a{}^b q_b\, ,
\end{equation}
where we refer to the direction aligned with $s$ as the longitudinal direction, while the quantities associated with $\mathcal{S}$ are referred to as transverse (or sheet) components. Hence, $p_{\parallel}$ and $\mathcal Q$ are longitudinal quantities, whereas $p_{\perp}$ and $\mathcal Q_a$ denote transverse ones.

Then the anisotropic stress tensor admits the canonical $(1{+}1{+}2)$~decomposition
\begin{equation}
  \pi_{ab}
  =
  \widetilde{\pi}
  \left(
  s_a s_b-\frac12\gamma_{ab}
  \right)
  +
  2\widetilde{\pi}_{(a}s_{b)}
  +
  \widetilde{\pi}_{ab} \, ,
  \label{eq:pi112-decomposition}
\end{equation}
where
\begin{align}
  \widetilde{\pi} &\defeq s^a s^b\pi_{ab}=
  \frac23\left(p_{\parallel}-p_{\perp}\right),
  \\
  \widetilde{\pi}_a &\defeq s^d\,\gamma_a{}^c\pi_{cd},
  \\
  \widetilde{\pi}_{ab}
  &\defeq
  \left(
  \gamma_{a}{}^c\gamma_{b}{}^d
  -
  \frac12\gamma_{ab}\gamma^{cd}
  \right)\pi_{cd}.
\end{align}

We can then summarize the results of this section as follows:
\begin{prop} \label{prop1}
Let ${\rm T} \in \Gamma ({\rm Sym}^2 (T^\ast U))$ be a symmetric rank-$(0,2)$ tensor, meant to describe an effective stress-energy tensor. Then its full $(1{+}1{+}2)$ decomposition reads
\begin{equation}\label{eq:T112-full-decomposition}
    T_{ab}= 
    \rho \, u_a u_b
  +
  p_{\parallel} \, s_a s_b
  +
  p_{\perp} \, \gamma_{ab}
  +
  2\mathcal Q \, u_{(a}s_{b)}
  +
  2u_{(a} \, \mathcal Q_{b)}
  +
  2\widetilde{\pi}_{(a} \, s_{b)}
  +
  \widetilde{\pi}_{ab}\, ,
\end{equation}
in abstract index notation.
\end{prop}

\section{Spacelike scalar-field gradient}
\label{sec:scalargrad}
%
In this section, we first specialize the analysis of the kinematic quantities introduced in Sec.~\ref{sec:kinematics} to the case of the spacelike direction aligned with the gradient of a scalar field. Then, we perform the imperfect-fluid decomposition, discussed in Sec.~\ref{sec:decomposition}, for Horndeski gravity adapted to the case of a scalar field with a spacelike gradient.

\subsection{Kinematics of the effective scalar-field fluid} \label{sec:scalar}

Let $\phi \in C^3(U, \mathbb{R})$ be a scalar field such that $g^{ab}\nabla_a\phi\nabla_b\phi > 0$. We employ the standard Horndeski convention 
\begin{equation}
  X\defeq-\frac12\nabla_a\phi\nabla^a\phi \, ,
\end{equation}
which for a spacelike scalar-field gradient yields $X<0$. 

Before {\em locking the spatial direction to the scalar gradient}, a general $(1{+}1{+}2)$ split with arbitrary $s$ would give
\begin{equation}
  \nabla_a\phi
  =
  -\dot{\phi} \, u_a
  +\widehat{\phi} \, s_a
  +\gamma_{a}{}^{b}\nabla_b\phi \,.
\end{equation}
Using the spacelike nature of the scalar gradient, we adapt the preferred spatial direction as
\begin{equation}
  s_a=\varepsilon\,\frac{\nabla_a\phi}{\sqrt{-2X}}\,,
  \qquad
  \varepsilon=\pm1\,.\label{eq:scalar_director}
\end{equation}
Borrowing terminology from the continuum description of anisotropic media, we refer to the resulting preferred spatial direction $s$ as the \emph{scalar director}. The sign $\varepsilon$ fixes its orientation. On any connected branch, it is taken to be constant; in the static spherically symmetric sector below, it will be chosen so that $s$ points toward increasing $r$.

Accordingly, by construction,
\begin{equation}
  \nabla_a\phi=\varepsilon\sqrt{-2X}\,s_a\,,
  \qquad
  \widehat{\phi}=s^a\nabla_a\phi=\varepsilon\sqrt{-2X}\,.
  \label{eq:gradphi-spacelike}
\end{equation}

We shall refer to the transformation $s^a\to-s^a$ as \emph{director reversal}. A quantity $F[s]$ is said to be \emph{director-reversal even} or \emph{director-reversal odd} according to whether
\begin{equation*}
F[-s]=+\,F[s]
\qquad\text{or}\qquad
F[-s]=-\,F[s],
\end{equation*}
respectively.

We can now prove the following results:
\begin{prop}
The following identities hold:
\begin{align}
&\nabla_a X=-\widehat{\phi}\,\nabla_a\widehat{\phi}=-\dot{X}u_a+\widehat{X}s_a+2X\,\widehat{s}_a\,, \label{eq:gradX}\\
&\gamma_{a}{}^{b}\nabla_b X=2X\,\widehat{s}_a \, , \label{eq:Xsheet} \\
&\widetilde{\omega} = 0 \, . \label{eq:sheet-vorticity}
\end{align}
\end{prop}
\begin{proof}
Using the decomposition~\eqref{eq:ds-decomposition}, we can write the scalar field Hessian in the following form
\begin{align}
  \nabla_b\nabla_a\phi
  &=\nabla_{b}\left(\widehat{\phi}\,s_{a}\right)\nonumber\\
  &=\left(\nabla_{b}\widehat{\phi}\right)s_{a}+\widehat{\phi}\,\nabla_{b}s_{a}\nonumber\\
  &=\left(\nabla_b\widehat{\phi}\right)s_a+\widehat{\phi}\left(-\dot{s}_{a} u_b+\widehat{s}_a s_b+\frac12\widetilde{\Theta}\,\gamma_{ab}+\widetilde{\omega}\,\epsilon_{ab}+\widetilde{\sigma}_{ab}+u_{a}\chi_{b}\right).
  \label{eq:hessian-spacelike}
\end{align}
Since $\nabla$ denotes the Levi-Civita connection, the torsionless property implies $\nabla_{[a}\nabla_{b]}\phi=0$, that projected on the spatial sheet $\mathcal{S}$ yields  
\begin{equation}
0=\gamma_b{}^c\gamma_a{}^d\nabla_{[c}\nabla_{d]}\phi=\widehat{\phi}\,\widetilde{\omega}\,\epsilon_{ab}\quad\Longrightarrow\quad\widetilde{\omega}=0\, ,
\end{equation}
i.e., the {\em spatial twist} vanishes identically.

Now, contracting the torsionless condition with $s^a$ and projecting onto the sheet gives
\begin{align}
0
&=
2\gamma_a{}^c s^d
\nabla_{[c}\nabla_{d]}\phi
\nonumber\\
&=
2\gamma_a{}^c s^d\left[\left(\nabla_{[c}\,\widehat{\phi}\right)s_{d]}+\widehat{\phi}\,\nabla_{[c}s_{d]}\right]
\nonumber\\
&=\gamma_a{}^c\,\nabla_c\widehat{\phi}-\widehat{\phi}\,\gamma_a{}^c s^d\nabla_d s_c
\nonumber\\
&=\gamma_a{}^c\,\nabla_c\widehat{\phi}-\widehat{\phi}\,\widehat{s}_a 
\nonumber\\
& =-\frac{\gamma_a{}^b\,\nabla_b X - 2X\,\widehat{s}_a}{\widehat\phi} \quad\Longrightarrow\quad \gamma_a{}^b\,\nabla_b X=2X\,\widehat{s}_a\,.
\end{align}
Then, Eq.~\eqref{eq:gradX} comes immediately from~\eqref{eq:Xsheet},~\eqref{eq:sheet-vorticity}, the definition $\nabla_a\phi=\widehat{\phi}\,s_a$, and the torsionless property of the connection.
\end{proof}
Evaluating the trace of Eq.~\eqref{eq:hessian-spacelike}, the d'Alembertian of the scalar field turns into
\begin{align}
  \Box\phi
  \defeq \nabla_a\nabla^a\phi=
  \widehat{\phi}\left(\mathcal{A}+\widetilde{\Theta}\right)-\frac{\widehat{X}}{\widehat{\phi}}\,.
  \label{eq:boxphi}
\end{align}
Furthermore, let us consider the relation $u^{a}\nabla_{b}\nabla_{a}\phi=u^{a}\nabla_{a}\nabla_{b}\phi$:
\begin{align}
  &u^{a}\nabla_{b}\nabla_{a}\phi=\widehat{\phi}\left(\mathcal{A}{u}_{b}-\chi_b\right)\,,
  \\ &u^{a}\nabla_{a}\nabla_{b}\phi=\left(u^{a}\nabla_{a}\widehat{\phi}\right)s_{b}+\widehat{\phi}(\mathcal{A}u_b+\alpha_b)\,,
\end{align}
and therefore
\begin{equation}
  \chi_b = \frac{\dot X}{\widehat\phi^2}\,s_b - \alpha_b\,.\label{eq:chi-scalar-locked}
\end{equation}
This allows us to state the following important proposition:
\begin{prop}
In the scalar-locked case, i.e., when the spatial direction $s$ is aligned with the direction of the gradient $\nabla \phi$, $\chi$ is not an independent kinematical object. Its longitudinal part is fixed by the time variation of the kinetic scalar, while its sheet part is fixed by the sheet component of $\dot s$.
\end{prop}
Using the expression $\chi_b
  =
  \frac13\Theta s_b+\sigma_{bd}s^d-\omega_{bd}s^d$, this also gives equivalent scalar-locked constraints
\begin{equation}
  \frac13\Theta+\sigma_{ab}s^a s^b
  =
  \frac{\dot X}{\widehat\phi^2}\,,
\end{equation}
and
\begin{equation}
  \gamma_b{}^c
  \left(
  \sigma_{cd}s^d-\omega_{cd}s^d
  \right)
  =
  -\alpha_b \,.
\end{equation}

Finally, it is useful to introduce the scalar
\begin{equation}
\Xi\defeq-\ln\left|\frac{\widehat{\phi}}{\phi_\star}\right|\,,
\label{eq:Xi-definition}
\end{equation}
where $\phi_\star$ is a constant introduced to make the logarithm argument dimensionless. Its derivatives along $u^a$ and $s^a$ satisfy
\begin{equation}
\dot{\Xi}
\defeq u^a\nabla_a\Xi
=
\frac{\dot X}{\widehat{\phi}^{\,2}}\,,
\qquad
\widehat{\Xi}
\defeq s^a\nabla_a\Xi
=
\frac{\widehat X}{\widehat{\phi}^{\,2}}\,.
\label{eq:Xi-derivatives}
\end{equation}
Combining Eqs.~\eqref{eq:chi-scalar-locked} and~\eqref{eq:Xi-derivatives}, Eq.~\eqref{eq:sds-decomposition} becomes
\begin{equation}
s^b\nabla_b s_a
=
\widehat{s}_a+\dot{\Xi}\,u_a \, ,
\label{eq:spacetime-versus-spatial-bend}
\end{equation}
for the scalar-locked case. The spacetime curvature vector of the integral curves of $s$, i.e., $s^b\nabla_b s_a$, is the covariant analogue of the curvature vector of a curve, and its spatial projection gives the spatial bend $\widehat{s}$. While the full spacetime vector may contain an additional component along the timelike direction $u$, the two coincide only when $\dot{\Xi}=0$.

\subsection{Horndeski effective stress-energy tensor}
\label{sec:horndeski-effective-stress}

The viable Horndeski action is given by~\cite{Horndeski:1974wa, Giusti:2021sku} 
\begin{equation}
S\left[ g_{ab}, \phi \right] = \frac{1}{16\pi}\int d^4 x \sqrt{-g} \, \left( G_4(\phi)R+G_{2}(\phi,X)-G_{3}(\phi,X)\Box\phi\right) + S^\mathrm{(m)} \,. \label{eq:viable-horndeski-action}
\end{equation}
Throughout this section, we use the shorthand notations
\begin{equation}
  G_{i\phi}\defeq \frac{\partial G_i}{\partial\phi},
  \qquad
  G_{iX}\defeq \frac{\partial G_i}{\partial X},
  \qquad
  G_{4\phi\phi}\defeq \frac{d^2G_4}{d\phi^2}.
\end{equation}

Performing the variation of~\eqref{eq:viable-horndeski-action} with respect to the metric tensor $g^{ab}$ and the scalar field $\phi$, one obtains the respective field equations (see e.g.~\cite{Horndeski:1974wa,Giusti:2021sku,Miranda:2022wkz}),
\begin{align}
G_4  \, G_{ab} -\nabla_{a}\nabla_{b}G_4  
+ \left[ \Box G_4 -\dfrac{G_2  }{2} 
-\dfrac{1}{2} \, \nabla_{c} 
\phi\nabla^{c}G_{3} 
\right] g_{ab}& \nonumber\\
+ \frac{1}{2}\left[ G_{3X} \, 
\Box\phi -G_{2X} \right] 
\nabla_{a}\phi\nabla_{b}\phi 
+ \nabla_{(a}\phi \nabla_{b)}G_{3}&=8\pi T^\mathrm{(m)}_{ab} \,,\label{eq:horndeski-feq}
\end{align}
\begin{align}
 G_{4\phi}  R + G_{2\phi} +G_{2X}  
\Box\phi+\nabla_{c}\phi\nabla^{c}G_{2X} &\nn\\
-G_{3X} (\Box\phi)^2-\nabla_{c}\phi\nabla^c 
G_{3X}  \Box\phi-G_{3X} \nabla^{c} 
\phi\Box\nabla_c\phi&\nn\\
+G_{3X}  R_{ab}\nabla^{a} 
\phi\nabla^{b}\phi-\Box G_{3} -G_{3\phi} \Box\phi&=0 \label{eq:horndeski-eom}\,,
\end{align}
where 
$$T^{\rm (m)}_{ab}\defeq-\frac{2}{\sqrt{-g}}\frac{\delta S^{\rm(m)}}{\delta g^{ab}}$$
denotes the variational definition of the matter stress-energy tensor.

The Horndeski field equation~\eqref{eq:horndeski-feq} can be recast in the form of the Einstein equations,
\begin{equation}\label{eq:effective-einstein-eq}
    G_{ab}=8\pi T^{(\rm eff)}_{ab}\,,
\end{equation}
where
\begin{equation}
    T^{(\rm eff)}_{ab}\defeq\frac{T^{(\rm m)}_{ab}}{G_4}+T^{( \phi)}_{ab}\,, 
\label{eq:T-eff}
\end{equation}
\begin{equation}
    T^{(\phi)}_{ab}\defeq T^{(2)}_{ab}+T^{(3)}_{ab}+T^{(4)}_{ab},
    \label{eq:T-phi}
\end{equation}
and the individual contributions are (see e.g.~\cite{Giusti:2021sku, Miranda:2022wkz})
\begin{align}
    8\pi T^{(2)}_{ab}=\,&\frac{1}{2G_4} \left( G_{2X}\nabla_{a}\phi\nabla_{b}\phi+G_{2}\,g_{ab} \right)\,,\\
    &\nn\\
     8\pi T^{(3)}_{ab}=&\frac{1}{2G_4} \left( G_{3X} \nabla_{c} X  \nabla^{c} \phi  - 2 X G_{3\phi} \right) g_{ab} \nn\\
    \,& - \frac{1}{2G_4}\left( 2 G_{3\phi} + G_{3X} \Box \phi \right) \nabla _a  \phi \nabla _b \phi- \frac{G_{3X}}{G_4} \nabla_{(a} X \nabla_{b)} \phi\,,\\
    &\nn\\
     8\pi T^{(4)}_{ab}=\,&\frac{G_{4\phi}}{G_4}(\nabla_{a}\nabla_{b}\phi-g_{ab}\Box\phi)+\frac{G_{4\phi\phi}}{G_4}(\nabla_{a}\phi\nabla_{b}\phi+2X\,g_{ab})\,.
\end{align}

\subsubsection{Characteristic \texorpdfstring{$(1{+}1{+}2)$}{(1+1+2)} relations}
\label{subsec:characteristic-relations}

We can now specialize the results in Proposition~\ref{prop1} to Horndeski gravity. 
Let $(M,g, \phi)$ be a sufficiently regular solution of the Horndeski field equations~\eqref{eq:horndeski-feq} and~\eqref{eq:horndeski-eom} such that the applicability conditions of the $(1{+}1{+}2)$ split and the scalar-locked kinematics are well defined. Then, using Eqs.~\eqref{eq:T112-full-decomposition},~\eqref{eq:gradphi-spacelike}, and~\eqref{eq:gradX}, the components of the $(1{+}1{+}2)$ imperfect-fluid representation for viable Horndeski gravity, in the scalar-locked case, are provided in the following theorem.

\begin{theorem}
The scalar-locked $(1{+}1{+}2)$ imperfect-fluid decomposition of the effective stress-energy tensor of the Horndeski scalar-field fluid in~\eqref{eq:T-phi} is given by Eq.~\eqref{eq:T112-full-decomposition} with components:
\begin{equation}\label{eq:energy-density-rho}
    8\pi\,\rho^{(\phi)}= 
\widehat{\phi}\,\frac{G_{4\phi}}{G_4}\,\widetilde{\Theta}-\frac{\left(G_{4\phi}-XG_{3X}\right)}{G_4}\frac{\widehat X}{\widehat{\phi}}
-\frac12\frac{\left(G_2+4XG_{4\phi\phi}-2XG_{3\phi}\right)}{G_4}\,,
\end{equation}

\begin{equation}\label{eq:pressure-parallel}
    8\pi\, p_{\parallel}^{(\phi)}=
\frac{
G_2
-2XG_{2X}
+2XG_{3\phi}
}{
2G_4
}-{\widehat{\phi}}
\left(
\mathcal A+\widetilde{\Theta}
\right)
\frac{\left(G_{4\phi}-XG_{3X}\right)}{G_4}\,,
\end{equation}

\begin{equation}\label{eq:pressure-perp}
    8\pi\, p_{\perp}^{(\phi)}=\frac{G_2+4XG_{4\phi\phi}-2XG_{3\phi}}{2G_4}
-\frac{\widehat{\phi}}{2G_4}\left(2\mathcal A+\widetilde{\Theta}\right)G_{4\phi}+\frac{\widehat X}{\widehat{\phi}}\frac{\left(G_{4\phi}-XG_{3X}\right)}{G_4}\,,
\end{equation}

\begin{equation}\label{eq:q-scalar}
    8\pi\, \mathcal{Q}^{(\phi)}=\frac{\dot X}{\widehat\phi}\,\frac{\left(G_{4\phi}-XG_{3X}\right)}{G_4}\,,
\end{equation}

\begin{equation}\label{eq:q-vector}
     8\pi\,\mathcal{Q}_a^{(\phi)}=- \widehat{\phi}\,\frac{G_{4\phi}}{G_{4}}\alpha_{a}\,,
\end{equation}

\begin{equation}\label{eq:pi-vector}
    8\pi\,\widetilde{\pi}_a^{(\phi)}={\widehat\phi}\,\frac{\left(G_{4\phi}-XG_{3X}\right)}{G_4}\widehat{s}_a\,,
\end{equation}

\begin{equation}\label{eq:pi-tensor}
    8\pi\,\widetilde{\pi}_{ab}^{(\phi)}={\widehat\phi}\,\frac {G_{4\phi}}{G_4}\,
\widetilde{\sigma}_{ab}\,.
\end{equation}
\end{theorem}

\begin{proof}
This result is obtained through the application of the various projections of the stress-energy tensor~\eqref{eq:T-phi} onto $u$, $s$, and $\mathcal{S}$, as detailed in Sec.~\ref{sec:decomposition}, where $s$ is taken as in~\eqref{eq:scalar_director}.
\end{proof}

Now, let us define the following scalar constitutive coefficients:
$$
\mathcal U_\perp
\defeq
\frac{G_2+4XG_{4\phi\phi}-2XG_{3\phi}}{2G_4}\,,
\qquad
\mathcal U_\parallel
\defeq
\frac{G_2-2XG_{2X}+2XG_{3\phi}}{2G_4},
$$
$$
\M
\defeq
\widehat\phi\,\frac{G_{4\phi}}{G_4}\,,
\qquad
\B
\defeq
-\widehat\phi\,\frac{XG_{3X}}{G_4}\,.
$$
Then, the above equations can be written as
\begin{align}
8\pi\,\rho^{(\phi)}
&=
-\mathcal U_\perp
+\M\,\widetilde\Theta
-(\M+\B)\,\widehat{\Xi}\,,\label{eq:energy_density}
\\[0.3em]
8\pi\,p_\parallel^{(\phi)}
&=
\mathcal U_\parallel
-(\M+\B)\,(\mathcal A+\widetilde\Theta)\,,\label{eq:longitudinal_pressure}
\\[0.3em]
8\pi\,p_\perp^{(\phi)}
&=
\mathcal U_\perp
-\M\!\left(\mathcal A+\frac12\widetilde\Theta\right)
+(\M+\B)\,\widehat{\Xi}\,,\label{eq:transverse_pressure}
\\[0.3em]
8\pi\,\mathcal Q^{(\phi)}
&=
(\M + \B)\,\dot{\Xi}\,,\label{eq:scalar_heat_flux}
\\[0.3em]
8\pi\,\mathcal{Q}_a^{(\phi)}
&=
- \M\,\alpha_{a}\,,\label{eq:vect_heat_flux}
\\[0.3mm]
8\pi\,\widetilde\pi_a^{(\phi)}
&=
(\M+\B)\,\widehat s_a\,,\label{eq:bend}
\\[0.3em]
8\pi\,\widetilde\pi_{ab}^{(\phi)}
&=
\M\,\widetilde\sigma_{ab}\,,\label{eq:biaxial_splay}
\end{align}
where $\widehat\Xi=\frac{\widehat X}{\widehat\phi^2}$ and $\dot\Xi=\frac{\dot X}{\widehat\phi^2}$.
The pressure anisotropy is
\begin{align}
8\pi\,\Pi^{(\phi)}
&\defeq
8\pi\left(p_\parallel^{(\phi)}-p_\perp^{(\phi)}\right)
\nn\\
&=
\left(\mathcal U_\parallel-\mathcal U_\perp\right)
-
 \B\,\mathcal A
-
\left(\frac12\M+\B\right)\widetilde\Theta
-
(\M+\B)\,\widehat\Xi \,.
\label{eq:anis}
\end{align}
It is now important to point out that this fluid decomposition is inherently non-unique for the (spacelike) scalar-locked case. Indeed, the time direction $u$ is picked from the Lorentzian subbundle ${\rm Ker} ({\rm d} \phi) \equiv (\nabla \phi)^\perp$, with signature $(-,+,+)$. Hence, any local Lorentz boost onto ${\rm Ker} ({\rm d} \phi)$ yields a new, equally valid, timelike vector that can act as the local 4-velocity of the fluid and remains orthogonal to $\nabla \phi$. In other words, this representation has a local ${\rm SO}^{\uparrow}(1,2)$ gauge symmetry that needs to be broken via kinetic closure, i.e., by choosing a specific physical frame for the decomposition.

Another important observation is that the ``constitutive'' equations~\eqref{eq:energy_density}--\eqref{eq:biaxial_splay} do not allow for a natural general fluid analogy as in the case of a timelike gradient (see~\cite{Giusti:2021sku}). Furthermore, the (would-be Cauchy) anisotropic stress tensor~\eqref{eq:bend}--\eqref{eq:biaxial_splay} no longer matches that of a Newtonian fluid, in general, nor is the heat-flux density~\eqref{eq:scalar_heat_flux}--\eqref{eq:vect_heat_flux} Eckart-like. 

Note that \eqref{eq:energy_density}--\eqref{eq:anis} confirm, and extend to a more general class of theories, some of the main conclusions of~\cite{Gergely:2026drx} in the spacelike scalar-field gradient case.

\section{Observer-dependent approach}
\label{sec:observer-dependent}

The characteristic relations~\eqref{eq:energy_density}--\eqref{eq:biaxial_splay} provide an observer-dependent description of the effective medium, which we often refer to as $\phi$-fluid, once a unit timelike vector field $u$ orthogonal to $\nabla\phi$ has been selected. However, this choice is not fixed by the scalar field itself because of the non-uniqueness of the fluid decomposition pointed out at the end of Sec.~\ref{subsec:characteristic-relations} for the scalar-locked case.

The following constitutive comparisons must therefore be understood relative to a specified timelike frame. We first examine the possible matching of the effective energy-flux channels with an anisotropic Eckart-like heat-flux law and consider the response of the vector and sheet deviatoric stresses to gradients of the scalar-locked preferred direction. Second, we highlight the common features between the constitutive structure of the effective scalar-field fluid for viable Horndeski gravity and the general theory of anisotropic media with a preferred spatial direction.

It is important to point out the limitations of the following analogies. The first, and main, limitation is the fact that this description is observer-dependent (indeed, we are required to pick a unit timelike vector field $u$ in ${\rm Ker} \, ({\rm d} \phi)$). Second, concerning more specifically the analogy with nematic fluids, in standard theory the fluid velocity and the director (i.e., the spatial unit vector field associated with the preferred direction) are independent material variables, whereas in our scalar-locked model $s$ is determined by the normalization of the scalar-field gradient and therefore it is not an independent degree of freedom. Hence, the map between the stress part of the effective scalar-field fluid for viable Horndeski gravity and the theory of nematic fluids is not one-to-one, yet this is not unexpected.
\subsection{Matching Eckart heat flux}
\label{subsec:eckart-matching}
The energy-flux channels of the scalar-locked effective fluid can be compared with relativistic heat conduction in a medium with a preferred spatial direction. Tentatively, to remain coherent with the original proposal of the thermodynamics of scalar-tensor gravity~\cite{Faraoni:2021lfc, Faraoni:2021jri}, one could investigate whether these channels admit an anisotropic Eckart-type heat-flux representation, once a timelike direction is selected for the fluid decomposition. To this end, we begin by recalling Eckart's law for an isotropic relativistic medium, namely,
\begin{equation}
q_a
=
-\kappa\,h_a{}^{b}
\left(
\nabla_b\T+\T\,\dot u_b
\right),
\label{eq:Eckart-heat-flux}
\end{equation}
where $\T$ is the local temperature and $\kappa\geq0$ is the thermal conductivity~\cite{Eckart:1940te}. As is apparent, Eckart's heat flux modifies Fourier's law by including an inertial contribution due to the fluid's acceleration. 

The general $(1{+}1{+}2)$ decomposition of the heat flux is given in~\eqref{eq:dec-q}, which, once applied to~\eqref{eq:Eckart-heat-flux}, yields longitudinal and sheet components of the Eckart flux. In detail, recalling that $\dot u_a=\mathcal As_a+\beta_a$, one finds
\begin{align}
\mathcal Q
&=
-\kappa
\left(
\widehat\T+\T\mathcal A
\right),
\label{eq:isotropic-Eckart-longitudinal-flux}
\\
\mathcal Q_a
&=
-\kappa
\left(
\gamma_a{}^b\nabla_b\T+\T\beta_a
\right) \, ,
\label{eq:isotropic-Eckart-sheet-flux}
\end{align}
for the Eckart flux in~\eqref{eq:Eckart-heat-flux}. Clearly, in the decomposition of an isotropic Eckart flux, the longitudinal and sheet components feature the same conductivity. The situation is notoriously different if one considers anisotropic media (see e.g.~\cite{Liu:2002Continuum} for a textbook introduction on the mathematical description of anisotropic materials).

In the scalar-locked case, it is reasonable to replace the conductivity coefficient with a conductivity tensor that yields different conductivity coefficients along the longitudinal direction and on the sheet. In other words, we can perform the replacement
\begin{equation}
\kappa h_a{}^b
\longrightarrow
\kappa_\parallel s_as^b
+
\kappa_\perp\gamma_a{}^b\,,
\qquad
\kappa_\parallel,\kappa_\perp\geq0\,.
\label{eq:uniaxial-conductivity-tensor}
\end{equation} 
into~\eqref{eq:Eckart-heat-flux}, ``anisotropizing'' the heat flux and obtaining an alternative law that reads 
\begin{equation}
q_a
=
-
\left(
\kappa_\parallel s_as^b
+
\kappa_\perp\gamma_a{}^b
\right)
\left(
\nabla_b\T+\T\dot u_b
\right).
\label{eq:uniaxial-fourier-eckart}
\end{equation}
In other words, since in the scalar-locked case there exists a preferred spatial direction $s$ given by the normalized scalar gradient, it is reasonable to replace the coefficient $\kappa$ with a uniaxial conductivity tensor ${\rm K}$ that reads
$$
{\rm K} \defeq  \kappa_\parallel \, s^\flat \otimes s^\flat
+
\kappa_\perp\gamma \, ,
$$
defined on $u^\perp$. Thermal conductivity tensors of this form are the simplest ones compatible with uniaxial media. Such anisotropic heat-flux laws are studied, for instance, in the context of uniaxial nematic media~\cite{Feireisl:2011, Feireisl:2012} and have also been discussed in the context of relativistic fluids with anisotropic heat conduction~\cite{Oettinger:1998}. It is worth observing that the thermal conductivity tensor ${\rm K}$ defined above, representing the replacement~\eqref{eq:uniaxial-fourier-eckart}, is director-reversal even by construction. This makes the heat flux in~\eqref{eq:uniaxial-fourier-eckart} inherently director-reversal even. We restrict our attention to this scenario. 

Now, projecting the anisotropic Eckart-like heat flux~\eqref{eq:uniaxial-fourier-eckart} along the longitudinal and sheet directions yields
\begin{align}
\mathcal Q
&=
-\kappa_\parallel
\left(
\widehat\T+\T\mathcal A
\right),
\label{eq:eckart-longitudinal-flux-general}
\\
\mathcal Q_a
&=
-\kappa_\perp
\left(
\gamma_a{}^b\nabla_b\T+\T\beta_a
\right).
\label{eq:eckart-sheet-flux-general}
\end{align}
Then, comparing these expressions with the characteristic relations of the Horndeski scalar-field fluid, specifically with Eqs.~\eqref{eq:scalar_heat_flux} and~\eqref{eq:vect_heat_flux}, yields
\begin{align}
(\M+\B)\,\dot\Xi
&=
-8\pi\kappa_\parallel
\left(
\widehat\T+\T\mathcal A
\right),
\label{eq:Horndeski-longitudinal-Eckart-matching-general}
\\
-\M\alpha_a
&=
-8\pi\kappa_\perp
\left(
\gamma_a{}^b\nabla_b\T+\T\beta_a
\right) \, .
\label{eq:Horndeski-sheet-Eckart-matching-general}
\end{align}
In other words, these latter equations determine the conditions on ${\rm K}$ and $\T$ that allow us to describe the effective heat flux of the Horndeski $\phi$-fluid in terms of a uniaxial Eckart-like heat transport.

Again, in analogy with previous literature~\cite{Faraoni:2021lfc, Faraoni:2021jri, Miranda:2022wkz}, we assume that the candidate effective temperature is a local scalar constructed from the scalar field and its first derivative, i.e.,
\begin{equation}
\T^{(\phi)}=\T^{(\phi)}(\phi,X)\,.
\label{eq:effective-temperature-phi-X}
\end{equation} 
Then, taking advantage of Eqs.~\eqref{eq:gradphi-spacelike} and~\eqref{eq:gradX}, the longitudinal and sheet temperature gradients read
\begin{align}
\widehat{\T}^{(\phi)}
&=
\widehat\phi\,\partial_{\phi}\T^{(\phi)}
+
\widehat X\,\partial_{X}\T^{(\phi)}\,,
\label{eq:temperature-longitudinal-gradient}
\\
\gamma_a{}^b\nabla_b\T^{(\phi)}
&=
2X\,\partial_{X}\T^{(\phi)}\,\widehat s_a\, ,
\label{eq:temperature-sheet-gradient}
\end{align}
and hence~\eqref{eq:Horndeski-longitudinal-Eckart-matching-general} and~\eqref{eq:Horndeski-sheet-Eckart-matching-general} can be rewritten as
\begin{align}
(\M+\B)\dot\Xi
&=
-8\pi\kappa_\parallel\T^{(\phi)}
\left[
\mathcal A
+
\widehat\phi\,\partial_\phi\ln\T^{(\phi)}
+
\widehat X\,\partial_X\ln\T^{(\phi)}
\right],
\label{eq:Horndeski-longitudinal-Eckart-matching-phi-X}
\\
-\M\alpha_a
&=
-8\pi\kappa_\perp\T^{(\phi)}
\left[
\beta_a
+
2X\,\partial_X\ln\T^{(\phi)}\,\widehat s_a
\right].
\label{eq:Horndeski-sheet-Eckart-matching-phi-X}
\end{align}
The dependence of the sheet force on $\widehat s_a$ is a direct consequence of the identity $\gamma_a{}^b\nabla_bX=2X\widehat s_a$, which holds for the scalar-locked case.

\begin{remark}
The subclass of our ansatz with $\T^{(\phi)}=\T^{(\phi)}(\phi)$ removes the explicit bend contribution from the sheet temperature gradient (i.e., the term proportional to $\widehat{s}$). The matching conditions then reduce to
\begin{align}
(\M+\B)\dot\Xi
&=
-8\pi\kappa_\parallel\T^{(\phi)}
\left(
\mathcal A
+
\widehat\phi\,\frac{\mathrm d\ln\T^{(\phi)}}{\mathrm d\phi}
\right),
\nonumber
\\
\M\alpha_a
&=
8\pi\kappa_\perp\T^{(\phi)}\beta_a\,.
\nonumber
\end{align}
Hence, unlike the timelike gradient case~\cite{Giusti:2021sku}, the longitudinal spatial gradient of a temperature that depends only on $\phi$ is still present since $\nabla\phi$ is aligned with $s$ rather than with $u$. $\triangleleft$
\end{remark}

\subsection{Anisotropic stress response}
\label{subsec:gradient-response}
It is now interesting to expand upon the analysis of the anisotropic stresses for the effective fluid via the standard methods of continuum mechanics for anisotropic fluids~\cite{Ericksen:1959, Ericksen:1960, Leslie:1966}. Specifically, we note that the vector~\eqref{eq:bend} and sheet trace-free tensor~\eqref{eq:biaxial_splay} characteristic relations, and more generally the would-be Cauchy stress tensor for the effective fluid, fit within the general framework for anisotropic media with a preferred direction (typically referred to as ``director'', and identified with our scalar-locked direction $s$). 

Tentatively, it would be possible to draw an analogy between the scalar-locked Horndeski effective fluid and some features of the Oseen--Frank theory for liquid crystals (see e.g.~\cite{Stewart:2004}), which corresponds to the continuum orientational elastic subclass of the Ericksen--Leslie theory for nematic fluids. However, the Oseen--Frank theory relies upon a free-energy density functional of the director field and its spatial gradients, which loses its natural physical meaning for our Horndeski effective fluid. Without dwelling on this point, we just point out that for the scalar-locked viable Horndeski fluid, the nematic-fluid-like relations read
\begin{equation}
\widetilde{\pi}^{(\phi)}_a
=
k_{\mathrm b}^{(\phi)}\widehat{s}_a\,,
\qquad
\widetilde{\pi}^{(\phi)}_{ab}
=
k_{\sigma}^{(\phi)}\widetilde{\sigma}_{ab}\,,
\label{eq:Horndeski-director-response}
\end{equation}
where
\begin{equation}
k_{\mathrm b}^{(\phi)}
\defeq
\frac{\M+\B}{8\pi}\,,
\qquad
k_{\sigma}^{(\phi)}
\defeq
\frac{\M}{8\pi}\, ,
\label{eq:director-response-coefficients}
\end{equation}
denote the effective \emph{director-gradient stress-response coefficients}.
\section{Observer-independent scalar-layer representation}
\label{sec:scalar-layer-representation}

Let $U \subseteq M$ be an open region of our spacetime in which $X<0$. Since $\nabla\phi$ is then everywhere spacelike, the level sets
\begin{equation}
\Sigma_{\phi_0}
\defeq
\left\{
x\in U\,\middle|\,\phi(x)=\phi_0
\right\}
\end{equation}
form a local foliation of $U$ by timelike hypersurfaces. This codimension-one scalar foliation also provides the geometrical starting point of the unitary-gauge effective-field-theory description of black-hole backgrounds with spacelike scalar profiles developed in~\cite{Mukohyama:2025jzk}. 
Our construction shares this geometrical starting point but uses it for a different purpose, namely to reorganize the full effective stress-energy tensor of viable Horndeski gravity in terms of the intrinsic and extrinsic geometry of these {\em scalar leaves} $\Sigma_{\phi_0}$, without selecting a particular timelike vector field within 
${\rm Ker}(\mathrm d\phi)$. The result is an observer-independent layered-media description of the effective stress-energy tensor.
\subsection{Geometry of the scalar leaves}
\label{subsec:scalar-foliation-geometry}

Let $\Sigma$ denote a generic scalar leaf. Then, the induced Lorentzian metric on each leaf is defined as
\begin{equation}
\ell_{ab} \defeq g_{ab}-s_a s_b = -u_a u_b+\gamma_{ab}\,, \qquad \ell_a{}^b s_b=0\,,
\label{eq:scalar-leaf-projector}
\end{equation}
which is independent of the particular unit timelike vector field $u\in {\rm Ker} \, (\mathrm d\phi)$ used to perform a $(1{+}1{+}2)$ decomposition, whereas the second equality in the definition of $\ell$ is given by the $(1{+}1{+}2)$ decomposition once such a timelike direction is selected. 

The {\em extrinsic curvature} of a scalar leaf, with the sign convention fixed by Eq.~\eqref{eq:scalar_director}, reads
\begin{equation}
\K_{ab} \defeq \ell_a{}^c\ell_b{}^d\nabla_c s_d\,,
\qquad
\K\defeq\ell^{ab}\K_{ab}\, ,
\label{eq:scalar-leaf-extrinsic-curvature}
\end{equation}
where $\K_{ab}$ denotes the second fundamental form of the scalar leaves, while $\K$ denotes its trace, measuring the variation of the unit normal $s$ to $\Sigma$ along directions tangent to the leaf $\Sigma$ (see e.g.~\cite{doCarmo:1992}). Equivalently, one has that $\mathcal L_s\ell_{ab}=2\K_{ab}$, where ${\cal L}_s$ denotes the Lie derivative along $s$.

Since $s \propto \nabla \phi$, the foliation is hypersurface orthogonal; furthermore, since $s$ is the unit normal to each leaf, the second fundamental form $\K_{ab}$ is symmetric. Indeed, it is easy to see that
$$
\K_{[ab]}= \ell_a{}^c\ell_b{}^d\nabla_{[c}s_{d]} = \ell_a{}^c\ell_b{}^d \nabla_{[c} \left(\frac{\varepsilon}{\sqrt{-2X}} \right)\,\nabla_{d]}\phi = 0 \, .
$$

We can now relate $\K_{ab}$ with the variables introduced in the $(1{+}1{+}2)$ split:
\begin{prop}
The $(1{+}1{+}2)$ decomposition of the second fundamental form of a scalar leaf $\Sigma$ reads
\begin{equation}
\K_{ab} = -\mathcal A u_a u_b -2u_{(a}\alpha_{b)} +\frac12\widetilde\Theta\,\gamma_{ab} +\widetilde\sigma_{ab} \, ,\label{eq:extrinsic-curvature}
\end{equation}
while its trace (essentially the mean curvature of the scalar leaf) is given by
\begin{equation}
\K=\mathcal A+\widetilde\Theta\,.
\label{eq:mean-curvature}
\end{equation}
\end{prop}
\begin{proof}
It follows directly from the scalar-locked kinematics. Indeed, taking advantage of 
Eqs.~\eqref{eq:ds-decomposition} and~\eqref{eq:chi-scalar-locked}, together with $\widetilde\omega=0$, the double tangential projection of $\nabla_b s_a$ gives exactly~\eqref{eq:extrinsic-curvature}. Then taking the trace of~\eqref{eq:extrinsic-curvature} immediately yields~\eqref{eq:mean-curvature}.
\end{proof}

We can now provide a geometrical interpretation for the formal quantity $\Xi \defeq - \ln|\widehat{\phi}/\phi_\star|$ defined in Sec.~\ref{sec:scalar}. Let $\lambda$ denote the parameter along the integral curve tangent to $s$ starting on the leaf $\Sigma$. Then,
$$
\frac{{\rm d} \phi}{{\rm d} \lambda} = s^a \nabla_a \phi = \widehat{\phi} \, ,
$$
denotes the variation of $\phi$ along this curve. Hence, a neighboring leaf $\Sigma _{\phi_0 + {\rm d} \phi}$ is located at a parameter distance ${\rm d} \lambda = {\rm d} \phi / \widehat{\phi}$ from $\Sigma_{\phi_0}$. Hence ${\rm d} \lambda = {\rm sign} (\widehat{\phi} ) \,{\rm d} \phi / |\widehat{\phi}| = {\rm sign} (\widehat{\phi} ) \,{\rm e}^{\Xi} \,{\rm d} \phi / |\phi_\star|$. Therefore, $\Xi$ is related to the increment of the parameter along the curve, i.e., the local leaf spacing in terms of $\lambda$, for an infinitesimal increment of the scalar field. Note that, given a local coordinate chart $\{x^\mu\}$, the integral curve of $s$ satisfies ${\rm d}x^{\mu}/{\rm d} \lambda = s^\mu$ and $s^\mu s_\mu =1$, hence ${\rm d} s ^2 = g_{\mu\nu} (x^{\alpha}(\lambda)) \, {\rm d} x^\mu (\lambda) \, {\rm d} x^\nu (\lambda) = {\rm d} \lambda ^2$, namely $\lambda$ denotes the proper distance along the integral curve. 
Therefore, given two infinitesimally close leaves $\Sigma _{\phi_0}$ and $\Sigma _{\phi_0 + {\rm d} \phi}$ in the foliation, then ${\rm e}^{\Xi}$ essentially measures the proper distance between the two leaves. In a similar manner, we have that $\widehat\Xi$ measures the variation of this spacing across the foliation.
\subsection{Decomposition of the stress-energy tensor}
\label{subsec:foliation-adapted-stress}
We can now decompose the effective stress-energy tensor taking advantage of this scalar-layer representation, i.e.,
\begin{equation}
T^{(\phi)}_{ab}
=
 p_{s}\,s_a s_b
+2s_{(a}j_{b)}
+\tau_{ab}\,,
\label{eq:leaf-stress-decomposition}
\end{equation}
where
\begin{equation}
 p_{s}
\defeq
T^{(\phi)}_{ab}s^a s^b\,,
\qquad
j_a
\defeq
\ell_a{}^cT^{(\phi)}_{cd}s^d\,,
\qquad
\tau_{ab}
\defeq
\ell_a{}^c\ell_b{}^dT^{(\phi)}_{cd}\,.
\label{eq:leaf-stress-definitions}
\end{equation}
In this decomposition, $ p_{s}$ is the pressure normal to the leaves, $j \in \Gamma (T^\ast \Sigma)$ is the leaf-tangential traction (satisfying $j_a \, s^a=0$), and $\tau \in \Gamma (T^\ast \Sigma \otimes T^\ast \Sigma)$ is the intrinsic scalar-leaf stress-energy tensor (satisfying $\tau_{ab} \, s^b=0$). These quantities can be easily expressed in terms of the $(1{+}1{+}2)$ variables as
\begin{align}
 p_{s}
&=p^{(\phi)}_{\parallel}\,,\label{eq:normal_stress}
\\
j_a&=\mathcal 
Q^{(\phi)}u_a+\widetilde\pi^{(\phi)}_a\,,\label{eq:normal_tangential_projection}\\
\tau_{ab}&=\rho^{(\phi)}u_a u_b+p^{(\phi)}_{\perp}\gamma_{ab}+2u_{(a}\mathcal Q^{(\phi)}_{b)}+\widetilde\pi^{(\phi)}_{ab}\,.
\label{eq:leaf-intrinsic_stress-energy_tensor}
\end{align}

We then have that
\begin{theorem}
For the effective stress-energy tensor of viable Horndeski gravity, the following relations hold
\begin{align}
8\pi\tau_{ab} &= \M \left(\K_{ab}-\K\ell_{ab}\right) + \mathcal P_{\ell}\,\ell_{ab}\,, \label{eq:leaf-intrinsic-stress} \\
\mathcal P_{\ell}
&\defeq
\mathcal U_{\perp}
+
(\M+\B)\widehat\Xi\, ,
\label{eq:leaf-isotropic-coefficient}\\
8\pi p_{s} &= \mathcal U_{\parallel} - (\M+\B)\K\, ,\label{eq:leaf-normal-pressure} \\
8\pi j_a &= (\M+\B) \left( \dot\Xi\,u_a+\widehat s_a \right)= (\M+\B) \,s^{b}\nabla_b s_a\,,
\label{eq:leaf-tangent-traction-112}
\end{align}
in the scalar-layer representation.
\end{theorem}
\begin{proof}
First, let us observe that 
\begin{equation}
\K_{ab}-\K\ell_{ab}
=
\widetilde\Theta\,u_a u_b -2u_{(a}\alpha_{b)} -\left(\mathcal  A+\frac12\widetilde\Theta\right)\gamma_{ab} +\widetilde\sigma_{ab}\,. \label{eq:trace-adjusted-extrinsic-curvature}
\end{equation}
Substituting the characteristic Horndeski relations~\eqref{eq:energy_density}--\eqref{eq:biaxial_splay} into~\eqref{eq:normal_stress} and~\eqref{eq:leaf-intrinsic_stress-energy_tensor} directly yields \eqref{eq:leaf-intrinsic-stress} and~\eqref{eq:leaf-normal-pressure} following some simple algebraic manipulations. Similarly, from the characteristic equations~\eqref{eq:scalar_heat_flux} and~\eqref{eq:bend}, the mixed projection in~\eqref{eq:normal_tangential_projection} yields the first equality in~\eqref{eq:leaf-tangent-traction-112}, while the last equality in~\eqref{eq:leaf-tangent-traction-112} follows from~\eqref{eq:spacetime-versus-spatial-bend}.
\end{proof}

Equations~\eqref{eq:leaf-intrinsic-stress} and~\eqref{eq:leaf-normal-pressure} show that $\M$ controls the contribution of the extrinsic curvature of the leaf to the intrinsic leaf stress, whereas $\M+\B$ gauges the response of the normal pressure to the mean curvature.

Let $\mathcal{D}$ be the Levi--Civita connection on $(\Sigma, \ell)$. Then, for the scalar $\Xi$, since $X=-\widehat\phi^{\,2}/2$, we have that
\begin{align}
\mathcal D_a\Xi &\defeq\ell_a{}^b\nabla_b\Xi = -u_a\dot\Xi+\gamma_a{}^b\nabla_b\Xi
\nonumber\\
& = -u_a\dot\Xi-\frac{1}{2X}\gamma_{a}{}^b\nabla_b X
\nonumber\\
& = -u_a\dot\Xi-\widehat s_a\,,
\end{align}
where we took advantage of the identity~\eqref{eq:Xsheet}. Then, from Eq.~\eqref{eq:leaf-tangent-traction-112} we have that
\begin{equation}
8\pi j_a = -(\M+\B)\mathcal D_a\Xi\,.
\label{eq:leaf-tangent-traction-gradient}
\end{equation}
Since $\Xi = -\ln\left|{\widehat{\phi}}/{\phi_\star}\right|$, if some suitable regularity conditions hold for $\phi$, then we have that (due to the inverse function theorem) $\widehat\phi$ and $X$ are locally determined by $\Xi$: $\widehat{\phi} = \varepsilon\,|\phi_\star|\,e^{-\Xi}$ and $X=-\frac12 \phi_\star^{\,2}\,e^{-2\Xi}$. That said, if we select a branch of this local inversion with a fixed orientation $\varepsilon$, then $\M+\B$ can locally be regarded as a function of $(\phi,\Xi)$. Since, by construction, we have that $\mathcal D_a\phi=0$ on each scalar leaf, then defining 
\begin{equation}
\Upsilon(\phi,\Xi) \defeq -\int_{\Xi_0}^{\Xi} 
\left[ \M(\phi,\zeta)+\B(\phi,\zeta) \right]\mathrm d\zeta\, ,
\label{eq:traction-potential-definition}
\end{equation}
we find that
\begin{equation}
8\pi j_a=\mathcal D_a\Upsilon\,. \label{eq:traction-potential}
\end{equation}
If this holds in the neighborhood of any point $p \in \Sigma$, then we have that $\forall p\in\Sigma \, $ $\exists V \subseteq \Sigma$ open neighborhood of $p$ and $\exists \Upsilon_V$ sufficiently smooth function such that $8\pi j_a |_{V}=\mathcal D_a\Upsilon_V$. In other words, the leaf-tangential traction is locally exact on every scalar leaf and satisfies
\begin{equation}
\mathcal D_{[a}\,j_{b]}=0\, ,\label{eq:traction-integrability}
\end{equation}
which gives a local integrability property of the effective Horndeski fluid.

This discussion can be summarized in the following:
\begin{theorem}
Under the above local invertibility assumption, $j$ is closed and locally exact on each scalar leaf.
\end{theorem}

To conclude the section we can observe that combining Eqs.~\eqref{eq:leaf-intrinsic-stress},~\eqref{eq:leaf-normal-pressure}, and~\eqref{eq:leaf-tangent-traction-gradient} allows us to rewrite effective Horndeski scalar-field stress-energy tensor as
\begin{align}
8\pi T^{(\phi)}_{ab}
={}&
\left[\mathcal U_{\parallel} -(\M+\B)\K \right]s_a s_b \nonumber\\
&-2(\M+\B)  s_{(a}\mathcal D_{b)}\Xi 
+\M \left(\K_{ab}-\K\ell_{ab}\right) \nonumber\\
&+\left[\mathcal U_{\perp} +(\M+\B)\widehat\Xi \right]\ell_{ab}\, , \label{eq:full-scalar-layer-stress}
\end{align}
and, based on our discussion above, 
the mixed term can also be written as
\begin{equation}
-2(\M+\B)s_{(a}\mathcal D_{b)}\Xi
=
2s_{(a}\mathcal D_{b)}\Upsilon \, .
\end{equation}
Clearly, all quantities in Eq.~\eqref{eq:full-scalar-layer-stress} are determined by the scalar foliation and the Horndeski functions. It is also worth noting that $\K_{ab}$, $\K$, $\M$, $\B$, and $\widehat\Xi$ are director-reversal odd, while $\ell_{ab}$, $\Xi$, and $\mathcal D_a\Xi$, and consequently the full~\eqref{eq:full-scalar-layer-stress}, are director-reversal even.

\subsection{Causal structure and preferred observers}
\label{subsec:causal-preferred-frames}
The observer-independent traction $j$ also clarifies which parts of its $(1{+}1{+}2)$ decomposition~\eqref{eq:normal_tangential_projection} are genuine and which are frame dependent. 

First, from~\eqref{eq:normal_tangential_projection} and~\eqref{eq:leaf-tangent-traction-gradient} we can compute
\begin{equation}
 j^2 \defeq  j^aj_a
 =-
 \left(\mathcal Q^{(\phi)}\right)^2
 +\widetilde\pi_a^{(\phi)}\widetilde\pi_{(\phi)}^a
 =\frac{(\M+\B)^2}{(8\pi)^2}
 \mathcal D_a\Xi\mathcal D^a\Xi \,.
\label{eq:traction-norm-main}
\end{equation}
From~\eqref{eq:leaf-tangent-traction-gradient} we know that if $\M+\B\neq0$, then $j_a$ and $\mathcal D_a\Xi$ determine the same leaf-tangential direction (whenever non-vanishing; otherwise they are both identically zero), with relative orientation fixed by the sign of $\M+\B$. Hence, they have the same causal character. When $\M+\B=0$, the traction vanishes independently of $\mathcal D_a\Xi$.

Notably, a timelike traction admits a unique future-directed ``traction-rest observer'' for which $\widetilde\pi_a^{(\phi)}=0$; a spacelike traction admits a one-parameter family of observers for which $\mathcal Q^{(\phi)}=0$; and a nonzero null traction admits neither. If $j_a=0$, both projections vanish in every scalar-adapted frame. The complete (pointwise) classification is given in Appendix~\ref{app:scalar-adapted-frames}.

The intrinsic leaf stress allows for the selection of a second class of preferred observers. Raising one index in Eq.~\eqref{eq:leaf-intrinsic-stress} gives 
\begin{equation}
8\pi\tau^a{}_b
=
\M\K^a{}_b
+\left(\mathcal P_{\ell}-\M\K\right)\ell^a{}_b \,.
\label{eq:leaf-stress-mixed-form-main}
\end{equation}
Hence, if $\M\neq0$, then $\tau^a{}_b$ and $\K^a{}_b$ have the same eigendirections. 

For $\M\neq0$, any timelike eigenvector of $\K^a{}_b$ is simultaneously a principal-curvature observer and a leaf-Landau observer. The former terminology is related to the fact that eigenvectors of the second fundamental form determine the principal curvature directions, while the latter is adopted in analogy with the standard Landau energy frame~\cite{Andersson:2020phh} (a timelike eigenvector of the fluid's stress-energy tensor).
However, the existence of such an observer is not guaranteed. Indeed, a self-adjoint operator on a Lorentzian leaf does not necessarily have a timelike eigenvector, as established in the Lorentzian algebraic classification of symmetric tensors~\cite{Santos:1995kt}.

The full four-dimensional scalar-adapted Landau criterion, for $\M\neq0$, reads
\begin{equation}
\K^a{}_b u^b\propto u^a\,,
\qquad
u^a j_a=0\,,
\label{eq:full-scalar-adapted-Landau-criterion-main}
\end{equation}
where the second condition removes the longitudinal flux. The detailed compatibility of traction-adapted, principal-curvature, leaf-Landau, and full Landau frames is collected in Appendix~\ref{app:scalar-adapted-frames}.

\subsection{Scalar-layer constitutive structure}
\label{subsec:constitutive-structure}

The observer-independent tensor decomposition in~\eqref{eq:full-scalar-layer-stress} can be compared with a class of first-order scalar-layer stress tensors.\footnote{Here ``first order'' refers to the constitutive expansion: $(\phi,X)$ are treated as local state variables, while $\K_{ab}$, $\K$, $\widehat\Xi$, and $\mathcal D_a\Xi$ constitute the retained response variables. Terms involving additional derivatives of these quantities or intrinsic leaf-curvature tensors belong to higher-derivative classes.} 

First, let us construct the most general first-order observer-independent scalar-layer stress-energy tensor from first principles. Let $(\phi,X)$ be the independent local state variables. Director-reversal-even state functions depend only on them, while director-reversal-odd response coefficients may carry an overall multiplicative factor of $\widehat\phi=\varepsilon\sqrt{-2X}$. Proceeding in analogy with standard continuum mechanics (see e.g.~\cite{Liu:2002Continuum}), we have that, identifying $\K_{ab}$, $\K$, $\widehat\Xi$, and $\mathcal D_a\Xi$ as first-order response variables, the stress-energy tensor will take the following schematic form
$$
T_{ab} = T_{ab} (\phi,X,\K_{ab}, \K,\widehat\Xi, \mathcal D_a\Xi, \cdots) \, .
$$
One then performs a formal expansion (derivative expansion, see e.g.~\cite{Dubovsky:2011sj}) of the constitutive relation in the ``number of derivatives'' of the state fields, i.e.,
$$
T_{ab} = T_{ab}^{(0)} + T_{ab}^{(1)} + T_{ab}^{(2)} + \cdots \, .
$$
In our case, $\phi$ and $X$ are order-zero quantities, whereas $\K_{ab}$, $\K$, $\widehat\Xi$, and $\mathcal D_a\Xi$ are order-one.

Since we are interested in symmetric tensors that are local,  director-reversal even within the leaves and overall, and at most linear in the response quantities, we have that, at the derivative order considered here, the scalar-locked identity
\begin{equation}
s^b\nabla_b s_a=-\mathcal D_a\Xi
\label{eq:normal-acceleration-layer-gradient-main}
\,,
\end{equation}
shows that $\mathcal D_a\Xi$ exhausts the available director-reversal even leaf-tangent vector sector. Furthermore, the most general observer-independent tensor within this restricted class reads
\begin{align}
8\pi T^{(\mathrm{layer})}_{ab}
={}&
\left[
\mathcal U_{\parallel}
+a_1\,\K
+a_2\,\widehat\Xi
\right]s_a s_b
+2a_3\,s_{(a}\mathcal D_{b)}\Xi
\nonumber\\
&+
a_4\,\K_{ab}
+
\left[
\mathcal U_{\perp}
+a_5\,\K
+a_6\,\widehat\Xi
\right]\ell_{ab}
\label{eq:general-linear-scalar-layer-medium}
\,.
\end{align}

Under $s^a\rightarrow-s^a$, the quantities $\widehat\phi$, $\K_{ab}$, $\K$, and $\widehat\Xi$ change sign, whereas $\ell_{ab}$, $\Xi$, and $\mathcal D_a\Xi$ remain unchanged.  Director-reversal invariance therefore requires the six derivative-response coefficients to be director-reversal odd. Hence, these coefficients may be written as
$a_i=\widehat\phi\,\bar a_i(\phi,X)$, where the functions $\bar a_i$ are director-reversal even.

Comparing~\eqref{eq:general-linear-scalar-layer-medium} with~\eqref{eq:full-scalar-layer-stress} gives
\begin{equation}
\begin{gathered}
a_1=-(\M+\B)
\,,
\qquad
a_2=0
\,,
\qquad
a_3=-(\M+\B)
\,,
\\
a_4=\M
\,,
\qquad
a_5=-\M
\,,
\qquad
a_6=\M+\B
\, ,
\end{gathered}
\label{eq:Horndeski-general-coefficient-map}
\end{equation}
and hence the six, a priori independent, derivative-response coefficients satisfy
\begin{equation}
a_2=0
\,,
\qquad
a_5=-a_4
\,,
\qquad
a_3=a_1
\,,
\qquad
a_6=-a_1
\,,
\label{eq:Horndeski-constitutive-identities-main}
\end{equation}
and reduce to the two combinations $\M$ and $\M+\B$ for the scalar-locked viable Horndeski case. These identities characterize the viable Horndeski response within the restricted scalar-layer class defined in Eq.~\eqref{eq:general-linear-scalar-layer-medium}.

The scalar-layer representation established above could support a partial comparison with lamellar media, specifically with smectic-A liquid crystals; see e.g.~\cite{ClimentEzquerra:2013sma, Paget:2022}. However, this analogy turns out to be not particularly relevant to this work, and it is therefore deferred to future studies.

\section{Static spherically symmetric spacetime}\label{sec:static}
%
We now specialize our discussion to the static spherically symmetric spacetime in which the scalar gradient is radial and spacelike. Hence, the metric and scalar profile read
\begin{equation}
g = - A(r) \, {\rm d} t \otimes {\rm d} t + B(r) \, {\rm d} r \otimes {\rm d} r + r^2 \, g_{\mathbb{S}^2}
\, ,   \quad \phi = \phi(r) \, ,
\end{equation}
in the coordinate chart $(t,r,\theta , \varphi)$ adapted to a Killing vector field
\[
K^a=\left(\frac{\partial}{\partial t}\right)^a \, ,
\]
that is timelike throughout the static region under consideration, with $t$ denoting the Killing time, with $r$ denoting the areal radius, and with $g_{\mathbb{S}^2} = {\rm d} \theta \otimes {\rm d} \theta + \sin^2 \theta \, {\rm d} \varphi \otimes {\rm d} \varphi$ denoting the standard metric on the unit 2-sphere.

Note that we are restricting the discussion to the case of $\phi=\phi(r)$. Scalar profiles such as $\phi (t,r) = q \, t + \psi (r)$, which are known to be compatible with static geometries in some shift-symmetric models, are not included here. Such configurations would require a separate analysis because the scalar gradient would not be purely locked to the radial spacelike direction~\cite{Babichev:2013cya}. More generally, any tensor field compatible with the symmetry is Lie-dragged along the generators of the symmetry group.

\subsection{Kinematics}\label{subsec:kinematics-static}

The natural congruence of static observers is the one aligned with the timelike Killing vector, namely
\begin{equation}
  u^a\defeq\frac{K^a}{\sqrt{-K^cK_c}}=\frac{1}{\sqrt{A}}\left(\frac{\partial}{\partial t}\right)^a,
  \qquad
  u_a=-\sqrt{A}\, ({\rm d}t)_a\,.
\end{equation}
The corresponding spatial projector is therefore
\begin{equation}
 h_{ab}=g_{ab}+u_a u_b=B(r)\,\nabla_a r\,\nabla_b r+r^2 \, g^{\mathbb{S}^2}_{ab}\,.
\end{equation}

Let us assume that $\phi(r)$ is monotonic and write
\begin{equation}
  \nabla_\mu\phi= \phi'(r)\, ({\rm d}r)_\mu =\phi'(r)\,\delta_\mu^r\,,
  \qquad
  X=-\frac{\phi'^2}{2B}<0\,.
\end{equation}
We choose the orientation sign $\varepsilon$ so that $s$ points toward increasing $r$. Equivalently, on the branch under consideration, $\varepsilon=\operatorname{sign}\bigl(\phi'(r)\bigr)$. Then, using the definition of $s$, we obtain
\begin{equation}
  s^a=\frac{1}{\sqrt{B}}\left(\frac{\partial}{\partial r}\right)^a\,,
  \qquad
  s_a=\sqrt{B}\,({\rm d}r)_a\,,
  \qquad
  \widehat\phi=s^a \nabla_a\phi=\frac{\phi'(r)}{\sqrt{B}}=\varepsilon\sqrt{-2X}\,.
\end{equation}
Accordingly, the sheet projector turns into
\begin{equation}
  \gamma_{ab}=g_{ab}+u_a u_b-s_a s_b=r^2 \, g^{\mathbb{S}^2}_{ab}\,.
\end{equation}
We now show explicitly which kinematical quantities are zero. First,
\begin{equation}
  \Theta=\nabla_a u^a=\frac{1}{\sqrt{-g}}\partial_\mu\left(\sqrt{-g}\,u^\mu\right)=\frac{1}{\sqrt{-g}}\partial_t\left(\sqrt{-g}\,u^t\right)=0,
\end{equation}
since $u^a$ has only a time component and all metric functions are time-independent. Next, for spatial indices $i,j=r,\theta,\varphi$ one has
\begin{equation}
  \nabla_i u_j=\partial_i u_j-\Gamma^\mu{}_{ij}u_\mu=-\Gamma^t{}_{ij}u_t=0,
\end{equation}
because $u_i=0$ and, for the static diagonal metric above, $\Gamma^t{}_{ij}=0$. Therefore,
\begin{equation}
  h_a{}^c h_b{}^d\nabla_c u_d=0 \quad \Rightarrow \quad \sigma_{ab}=0=\omega_{ab}\,.
\end{equation}

Let us now turn to the kinematics of $s_\mu=\sqrt{B}\,\delta_\mu^r$. Its derivative along $u^a$ is
\begin{align}
  \dot{s}_\mu
  &=u^\alpha\nabla_\alpha s_\mu
  =u^t\nabla_t s_\mu
  =-u^t\Gamma^\alpha{}_{t\mu}s_\alpha
  =-\frac{1}{\sqrt{A}}\Gamma^r{}_{t\mu}\sqrt{B}.
\end{align}
Since the only non-vanishing component is obtained for $\mu=t$, with
\begin{equation}
  \Gamma^r{}_{tt}=\frac{A'}{2B},
\end{equation}
we get
\begin{equation}
  \dot{s}_\mu=-\frac{A'}{2\sqrt{AB}}\,\delta_\mu^t
  =\frac{A'}{2A\sqrt{B}}u_\mu.
\end{equation}
Comparing with $\dot{s}_a=\mathcal{A}u_a+\alpha_a$, this shows that
\begin{equation}
  \mathcal A=\frac{A'}{2A\sqrt{B}}\,,\qquad \alpha_\mu=0\,.
\end{equation}
Likewise,
\begin{align}
  \dot{u}_\mu
  &=u^\alpha\nabla_\alpha u_\mu
  =u^t\nabla_t u_\mu
  =-u^t\Gamma^\gamma{}_{t\mu}u_\gamma
  =\frac{1}{\sqrt{A}}\Gamma^t{}_{t\mu}\sqrt{A}
  =\frac{A'}{2A}\,\delta_\mu^r=\mathcal A s_\mu\,,
\end{align}
which is purely radial, or equivalently,
\begin{equation}
    \dot{u}_\mu=\nabla_\mu\ln\sqrt{A(r)}\,.
\end{equation}
Therefore
\begin{equation}
  \beta_\mu\defeq\gamma_\mu{}^\beta\dot{u}_\beta=0.
\end{equation}

For $\widehat{s}_a$ one finds
\begin{align}
  \widehat{s}_a
  &=s^bD_b s_a
  =s^b h_a{}^c\nabla_b s_c.
\end{align}
Since $s_c$ has only a radial component, the only potentially non-vanishing contribution is the radial one:
\begin{align}
  \widehat{s}_r
  &=s^r\left(\partial_r s_r-\Gamma^r{}_{rr}s_r\right)
  =\frac{1}{\sqrt{B}}\left(\frac{B'}{2\sqrt{B}}-\frac{B'}{2B}\sqrt{B}\right)=0,
\end{align}
and all the other components vanish identically. Hence $\widehat{s}_a=0$.

Next, let us consider the projected derivative of $s_a$ on the sheet. Since $s_a$ has no angular components,
\begin{equation}
  \gamma_\alpha{}^\mu\gamma_\beta{}^\nu\nabla_\mu s_\nu =\gamma_\alpha{}^\mu\gamma_\beta{}^\nu\left(\partial_\mu s_\nu-\Gamma^\sigma{}_{\mu \nu}s_\sigma\right)
  =-\sqrt{B}\,\gamma_\alpha{}^\mu\gamma_\beta{}^\nu\Gamma^r{}_{\mu\nu}.
\end{equation}
The only non-vanishing components are
\begin{equation}
  \nabla_\theta s_\theta=-\Gamma^r{}_{\theta\theta}s_r=\frac{r}{\sqrt{B}},
  \qquad
  \nabla_\varphi s_\varphi=-\Gamma^r{}_{\varphi\varphi}s_r=\frac{r\sin^2\theta}{\sqrt{B}},
\end{equation}
which can be written as
\begin{equation}
  \gamma_\alpha{}^\mu\gamma_\beta{}^\nu\nabla_\mu s_\nu=\frac{1}{r\sqrt{B}}\,\gamma_{\alpha \beta}\quad\Rightarrow\quad\widetilde{\Theta}=\frac{2}{r\sqrt{B}}\,,\quad \widetilde{\omega}=0=\widetilde{\sigma}_{\alpha \beta}\,.
\end{equation}
Finally, since $X=X(r)$,
\begin{equation}
  \dot X=u^\alpha\nabla_\alpha X=0\,,\quad\widehat X=s^\alpha\nabla_\alpha X=\frac{1}{\sqrt{B}}\frac{{\rm d}X}{{\rm d}r}
  =
\frac{1}{B^{3/2}}\left(-\phi'\phi''+\frac{\phi'^2 B'}{2B}\right).
\end{equation}
Hence, in a static spherically symmetric spacetime, the only non-vanishing scalar kinematical quantities are 
$$
\mathcal{A}=\frac{1}{\sqrt{B}} \frac{{\rm d}\ln\sqrt{A}}{{\rm d}r} \, , \quad 
\widetilde{\Theta}=\frac{2}{r\sqrt{B}} \, , \quad \mbox{and} \quad
\widehat X= \frac{1}{\sqrt{B}} \frac{{\rm d}X}{{\rm d}r} \, .
$$

This discussion can then be summarized as follows:
\begin{theorem} \label{th:kin-static}
    For a scalar-locked effective viable Horndeski fluid on a static, spherically symmetric spacetime, with the scalar field depending only on the radial coordinate, the associated kinematic quantities of the $(1{+}1{+}2)$ split read:
    $$
    \widehat s_\alpha=0,
    \quad 
    \widetilde{\sigma}_{\alpha \beta}=0,
    \quad
    \widetilde{\omega} = 0,
    \quad
    \alpha_\alpha=0 \, , \quad \beta_\alpha=0 \, , \quad
    \mathcal{A}=\frac{(\ln\sqrt{A})'}{\sqrt{B}}  \, , \quad \mbox{and} \quad 
    \widetilde{\Theta}=\frac{2}{r\sqrt{B}} \, , 
    $$
    with respect to the coordinate chart $(t,r,\theta , \varphi)$ and 
    where the prime denotes the derivative with respect to $r$. Furthermore, it holds that $\dot{X} =0$ and $\widehat{X} = X'/\sqrt{B}$.
\end{theorem}
\subsection{Scalar-layer form of the static medium}
\label{subsec:static-scalar-layer}
We can now easily use the results of Theorem~\ref{th:kin-static} to begin computing the full $(1{+}1{+}2)$ imperfect-fluid decomposition and scalar-layer decomposition of a scalar-locked effective viable Horndeski fluid on a static, spherically symmetric spacetime, with the scalar field depending only on the radial coordinate. Indeed, combining the results of Theorem~\ref{th:kin-static} with~\eqref{eq:q-scalar}--\eqref{eq:biaxial_splay} we immediately find a corollary to Theorem~\ref{th:kin-static}:
\begin{corollary}
    For a scalar-locked effective viable Horndeski fluid on a static, spherically symmetric spacetime, with the scalar field depending only on the radial coordinate, the heat-flux channels and the vector/tensor anisotropic channels of the general $(1{+}1{+}2)$ split of the effective stress-energy tensor read:
    $$
    \mathcal Q^{(\phi)}_\alpha=0\,,
    \quad
    \mathcal Q^{(\phi)} =0\,,
    \quad
    \widetilde{\pi}_\alpha^{(\phi)}=0,
    \quad \mbox{and} \quad
    \widetilde{\pi}_{\alpha \beta}^{(\phi)}=0 \, . 
    $$
\end{corollary}

This means that the effective scalar sector reduces to a purely uniaxial equilibrium medium, fully characterized by $\rho^{(\phi)}$, $p_\parallel^{(\phi)}$, and $p_\perp^{(\phi)}$. 

For the considered static spherically symmetric configurations with $\phi=\phi(r)$, the scalar leaves coincide locally with the timelike worldtubes $r=\mathrm{const.}$ (the latter being diffeomorphic to $\mathbb{R} \times \mathbb{S}^2$ in the static region). Furthermore, by staticity and spherical symmetry, we have that $\Xi=\Xi(r)$, and hence
\begin{equation}
\mathcal D_\alpha\Xi=0, \qquad j_\alpha=0\,.
\label{eq:static-traction-vanishes}
\end{equation}
Moreover, from Theorem~\ref{th:kin-static} we also have that the extrinsic curvature remains non-vanishing in general; indeed
\begin{equation}
\K_{\alpha \beta}
=
-\mathcal A u_\alpha u_\beta
+\frac12\widetilde\Theta\,\gamma_{\alpha \beta}\,,
\qquad
\K
=
\mathcal A+\widetilde\Theta\,.
\label{eq:static-scalar-leaf-curvature}
\end{equation}
Thus, the leaf-stress relation does not trivialize:
\begin{equation}
8\pi\tau_{\alpha \beta}
=
\M
\left(
\K_{\alpha \beta}-\K\ell_{\alpha \beta}
\right)
+
\mathcal P_{\ell}\ell_{\alpha \beta}\,.
\label{eq:static-leaf-stress}
\end{equation}
The static background can thus be viewed geometrically as a continuous foliation by concentric timelike layers with vanishing tangential traction, a generally non-vanishing normal stress, and an intrinsic stress depending linearly on their extrinsic curvature. While the bend and biaxial-splay channels vanish in exact static spherical symmetry, this curvature-dependent leaf stress remains non-trivial on the background.
Interestingly, the static observer is simultaneously a principal-curvature observer, a leaf-Landau observer, and a Landau observer of the complete effective scalar stress-energy tensor.

\subsection{Tolman--Ehrenfest equilibrium criterion}\label{subsec:tolman}

A natural way to introduce the temperature profile in a static gravitational field is through the Tolman--Ehrenfest criterion~\cite{Tolman:1930zza, Tolman:1930ona}. In the Eckart-like fluid analogy, equilibrium is characterized by the vanishing of the heat flux in the local rest frame,
\begin{equation}
q_a=0.
\end{equation}
In Eckart's first-order thermodynamics, this condition is encoded in the constitutive relation~\cite{Eckart:1940te}
\begin{equation}
q_a=-\kappa\,h_a{}^{b}\left(\nabla_b \T+\T\,\dot u_b\right),
\end{equation}
so that, for nonvanishing thermal conductivity, equilibrium requires
\begin{equation}
\label{eq:tolcrit}
h_a{}^{b}\nabla_b \T+\T\,\dot u_a=0.
\end{equation}
This is the local covariant form of the Tolman--Ehrenfest criterion, which states that temperature gradients are exactly balanced by the four-acceleration of the static observers.

In the static spherically symmetric case, the Tolman temperature depends only on $r$, and the acceleration is purely radial, $\dot u_\alpha =\nabla_\alpha \ln\sqrt{A(r)}$. Hence, from~\eqref{eq:tolcrit} it is easy to show that
\begin{equation}
\T(r)=\frac{\T_0}{\sqrt{A(r)}} \, ,
\end{equation}
although this result was expected in light of~\cite{Karolinski:2024ukr}.

Now, let us make explicit the link between Tolman temperature and $\mathcal A\defeq s^a\dot u_a$:
$$
\widehat{\ln\T}\defeq s^\alpha\nabla_\alpha\ln\T
=
\frac{(\ln\T)'}{\sqrt{B}} 
=
-\frac{1}{2\sqrt{B}}\frac{A'}{A}
=
-\mathcal A.
$$
Hence the radial acceleration of the static congruence is exactly the Tolman temperature gradient,
\begin{equation}
\mathcal A=-\widehat{\ln\T} \, .
\end{equation}
Note that the same radial equilibrium condition follows directly from the uniaxial law~\eqref{eq:eckart-longitudinal-flux-general} when $\mathcal Q=0$ and $\kappa_\parallel>0$. Moreover, let us stress that here $\T$ denotes the Tolman temperature measured by the static observers and is not, in general, identified with the candidate effective temperature $\T^{(\phi)}(\phi,X)$ introduced in Sec.~\ref{subsec:eckart-matching}.

Using this identity, the characteristic equations~\eqref{eq:energy_density}--\eqref{eq:transverse_pressure} can be rewritten in a form where the local equilibrium temperature enters explicitly:
\begin{align}
8\pi\,\rho^{(\phi)}
&=
-\mathcal U_\perp+\M\,\widetilde{\Theta}-(\M+\B)\,\widehat\Xi, \label{rho}
\\
8\pi\,p_{\parallel}^{(\phi)}
&=
\mathcal U_\parallel+(\M+\B)\,\widehat{\ln\T}-(\M+\B)\,\widetilde{\Theta}, \label{p1}
\\
8\pi\,p_{\perp}^{(\phi)}
&=
\mathcal U_\perp+\M\,\widehat{\ln\T}-\frac12\,\M\,\widetilde{\Theta}+(\M+\B)\,\widehat\Xi, \label{p2}
\end{align}
where
\begin{equation}
\widehat\Xi\defeq-\widehat{\ln\lvert\widehat\phi\rvert}\,,
\qquad
\widetilde{\Theta}=\frac{2}{r\sqrt{B}}=2 \widehat{\ln r}\,,
\qquad
\M\defeq\widehat\phi \,\frac{G_{4\phi}}{G_4}=\widehat{\ln G_{4}}\,,
\qquad
\B
\defeq
-\widehat\phi\,\frac{XG_{3X}}{G_4}\,.
\end{equation}
Accordingly, the effective scalar stress-energy tensor takes the uniaxial equilibrium form
\begin{equation}\label{eq:static_energy_tensor}
T^{(\phi)}_{\alpha \beta}
=
\rho^{(\phi)}u_\alpha u_\beta
+
p_{\parallel}^{(\phi)} s_\alpha s_\beta
+
p_{\perp}^{(\phi)} \gamma_{\alpha \beta},
\end{equation}
with coefficients as in~\eqref{rho}--\eqref{p2}.

This rewriting makes clear that the Tolman law fixes the local equilibrium temperature profile measured by the static observers in terms of the lapse function $A(r)$. Since the lapse is determined on shell by the Horndeski field equations, this profile depends implicitly on the Horndeski functions, the scalar configuration, and the boundary conditions. The Tolman relation alone, however, does not determine the local constitutive coefficients $\mathcal U_\parallel$, $\mathcal U_\perp$, $\M$, and $\B$, nor does it provide a universal relation between them and $\T$.

The static reduction shares the formal uniaxial stress structure of relativistic elastic equilibrium models of spherical matter~\cite{Karlovini:2002fc, Alho:2023ris}. At the same time, the coexistence of an equilibrium temperature profile with a static anisotropic stress response shares some similarities with relativistic thermoelastic frameworks such as in~\cite{Kijowski:1997mx}.

\section{Conclusions}

We have presented a thermomechanical framework for the investigation of viable Horndeski scalar hair with spacelike scalar-field gradients. Within the imperfect-fluid and thermodynamic approaches to scalar-tensor gravity, this was one of the missing pieces required to study configurations such as black holes.

In Sec.~\ref{sec:kinematics} we began by introducing the generalities of the kinematics of the $(1{+}1{+}2)$ split, while in Sec.~\ref{sec:decomposition} we described the imperfect-fluid representation for generic stress-energy tensors in the $(1{+}3)$ and $(1{+}1{+}2)$ decompositions.

In Sec.~\ref{sec:scalargrad} we specialized the general arguments of the previous sections to the case of a preferred spatial direction determined by a scalar field with spacelike gradient (i.e., the scalar-locked direction), in the context of viable Horndeski gravity. The result is a $(1{+}1{+}2)$ effective imperfect-fluid representation of the class, together with the corresponding effective constitutive laws. Note that the results of Sec.~\ref{sec:scalargrad} confirm and extend those of~\cite{Gergely:2026drx} to viable Horndeski gravity, within the spacelike-gradient case.

This $(1{+}1{+}2)$ imperfect fluid description turns out to be non-unique for the scalar-locked case. The reason is that the timelike direction is not directly determined by the scalar-field gradient, as is instead the case in the thermodynamics of scalar-tensor gravity~\cite{Faraoni:2021lfc, Faraoni:2021jri}; rather, it displays a local ${\rm SO}^{\uparrow}(1,2)$ gauge symmetry. On the one hand, this leads to an observer-dependent $(1{+}1{+}2)$ framework for the scalar-locked effective fluid of viable Horndeski gravity [Sec.~\ref{sec:observer-dependent}], that maps roughly onto the theory of anisotropic media with a preferred spatial direction, upon choosing a timelike direction $u \in {\rm Ker} \, ({\rm d} \phi)$. On the other hand, in Sec.~\ref{sec:scalar-layer-representation} we formulated an observer-independent framework that ultimately leads to a layered-media description of the effective stress-energy tensor based on the foliation of the considered spacetime region by constant-scalar leaves. This latter description is what we refer to as the scalar-layer representation.
More specifically, the scalar-layer representation shows that the intrinsic leaf stress contains a contribution proportional to the trace-adjusted extrinsic curvature of the scalar leaves, while the leaf-tangential traction is locally exact and can be written as the intrinsic gradient of a scalar potential. The causal character of this traction determines which of its observer-dependent projections can be eliminated by a change of scalar-adapted frame, while the intrinsic leaf stress provides geometrical criteria for the existence of principal-curvature, leaf-Landau, and full Landau observers. Moreover, within the restricted first-order scalar-layer constitutive class considered here, the six derivative-response coefficients reduce to the two combinations $\M$ and $\M+\B$.

In Sec.~\ref{sec:static} we then specialized the analysis to the case of static spherically symmetric spacetime with scalar hair which depends only on the areal radius. This is indeed the case of interest for many relevant solutions of scalar-tensor gravity, specifically in the context of black holes, wormholes, and naked singularities (see e.g.~\cite{Faraoni:2021nhi} and references therein). An interesting consequence of this specialization is that it allows us to connect the thermomechanical framework with the Tolman--Ehrenfest criterion for thermal equilibrium. 
Indeed, the heat-flux channels and the vector and sheet trace-free anisotropic stresses vanish, so that the effective scalar source reduces to a uniaxial equilibrium medium, while the scalar-layer intrinsic stress retains a nontrivial dependence on the extrinsic curvature of the scalar leaves. For the static observers, the Tolman--Ehrenfest condition yields the standard relation $\T(r)\sqrt{A(r)}=\mathrm{constant}$.

\section*{Acknowledgements}

\noindent M.S. and A.G. are supported by the Italian Ministry of Universities and Research (MUR) through the grant ``BACHQ: Black Holes and The Quantum'' (grant no. J33C24003220006) and by the INFN project FLAG. M.M. acknowledges support from the INFN  project MOONLIGHT-2. This work has been carried out in the framework of activities of the National Group of Mathematical Physics (GNFM, INdAM).

\appendix

\section{Pointwise classification of scalar-adapted frames}
\label{app:scalar-adapted-frames}

This appendix completes the analysis of Sec.~\ref{subsec:causal-preferred-frames} by giving the exhaustive pointwise classification of the distinguished scalar-adapted frames. For each class of frames, we reproduce the minimal set of identities needed for the classification, while referring to the main text for their derivation and geometrical interpretation.

\subsection{Causal classification of the leaf-tangential traction}
\label{app:causal-classification}

For any scalar-adapted observer, i.e., a future-directed unit timelike vector $u \in {\rm Ker} \, ({\rm d} \phi)$, the leaf-tangential traction decomposes as [cf.~\eqref{eq:normal_tangential_projection}]
\begin{equation}
j_a=\mathcal Q^{(\phi)}u_a+\widetilde\pi_a^{(\phi)} \,\,\, \Rightarrow \,\,\,
u^aj_a=-\mathcal Q^{(\phi)}.
\label{eq:traction-observer-decomposition-app}
\end{equation}
Its intrinsic gradient representation and squared norm are [cf.~\eqref{eq:leaf-tangent-traction-gradient} and~\eqref{eq:traction-norm-main}]
\begin{align}
8\pi j_a&=-(\M+\B)\mathcal D_a\Xi,
\label{eq:traction-gradient-app}
\\
j^2&=-\left(\mathcal Q^{(\phi)}\right)^2+\widetilde\pi_a^{(\phi)}\widetilde\pi_{(\phi)}^a.
\label{eq:traction-norm-app}
\end{align}
If $\M+\B\neq0$, $j_a$ and $\mathcal D_a\Xi$ span the same leaf-tangential one-dimensional subspace (when non-vanishing), with proportionality factor given by $-(\M+\B)/(8\pi)$. If instead $\M+\B=0$, then $j_a=0$ independently of $\mathcal D_a\Xi$, so the last row of Table~\ref{tab:traction-classification} applies. The causal character of $j_a$ therefore determines which of its frame-dependent projections can be removed by changing the timelike observer. The four distinct cases are collected in Table~\ref{tab:traction-classification}.

\begin{table}[!htbp]
\begingroup
\small
\setlength{\tabcolsep}{5pt}
\renewcommand{\arraystretch}{1.45}

\caption{Causal classification of the leaf-tangential traction $j_a$
and the different resulting frames.}
\label{tab:traction-classification}

\begin{tabularx}{\textwidth}
{@{}L{2.7cm}L{4.0cm}C{2.8cm}Y@{}}
\toprule
Causal type
&
Observer(s)
&
Relation
&
Projections
\\
\midrule

$j^2<0$
&
Unique future-directed traction-rest observer
$
\displaystyle u_{(j)}^a
=
\pm\frac{j^a}{\sqrt{-j^2}}
$
with the sign fixed by the time orientation
&
$u_{(j)}^a\parallel j^a$
&
$\widetilde\pi_a^{(\phi)}=0$ and
$\left|\mathcal Q^{(\phi)}\right|=\sqrt{-j^2}$.
\\
\addlinespace[0.5em]

$j^2>0$
&
One-parameter family of future-directed unit timelike observers
tangent to the leaf
&
$u^aj_a=0$
&
$\mathcal Q^{(\phi)}=0$ and
$\widetilde\pi_a^{(\phi)}=j_a$.
\\
\addlinespace[0.5em]

$j^2 =0$, $j_a\neq0$
&
No traction-rest or vanishing-longitudinal-flux observer
&
\textemdash
&
Neither $\mathcal Q^{(\phi)}$ nor $\widetilde\pi_a^{(\phi)}$ can be set to zero in any scalar-adapted frame, and $\widetilde\pi_a^{(\phi)} \widetilde\pi_{(\phi)}^a = \left(\mathcal Q^{(\phi)}\right)^2$. \\
\addlinespace[0.5em]

$j_a=0$
&
Every scalar-adapted observer
&
\textemdash
&
$\mathcal Q^{(\phi)}=0=\widetilde\pi_a^{(\phi)}$ identically.
\\

\bottomrule
\end{tabularx}
\endgroup
\end{table}

Under $s^a\rightarrow-s^a$, one has $j_a\rightarrow-j_a$, so the sign in $u_{(j)}^a$ must be chosen to preserve the same future-directed observer. The causal class and all existence statements are independent of this orientation convention. The classification is pointwise and may change across regions where $j^2=0$ or $j_a=0$. Thus, $\mathcal Q^{(\phi)}$ and $\widetilde\pi_a^{(\phi)}$ are observer-dependent projections of the single observer-independent traction $j_a$, rather than two independent intrinsic response channels.

\subsection{Principal-curvature and leaf-Landau frames}
\label{app:principal-landau}

The intrinsic leaf stress satisfies~\eqref{eq:leaf-intrinsic-stress}, i.e.,
\begin{equation}
8\pi\tau^a{}_b=\M\K^a{}_b+\left(\mathcal P_{\ell}-\M\K\right)\ell^a{}_b \, ,
\label{eq:leaf-stress-mixed-form-app}
\end{equation}
reported again here for convenience.
Since $\ell^a{}_b$ acts as the identity on vectors tangent to a scalar leaf, $\tau^a{}_b$ is an affine function of $\K^a{}_b$ if $\M\neq0$. The two operators therefore have the same eigendirections and the same eigenspaces associated with real eigenvalues. A \emph{principal-curvature observer} is a future-directed unit timelike eigenvector of $\K^a{}_b$, whereas a \emph{leaf-Landau observer} is independently defined by the vanishing of the intrinsic leaf-spatial energy flux, $\mathcal Q_a^{(\phi)}=0$, in analogy with the standard Landau energy frame~\cite{Andersson:2020phh}.

For an arbitrary scalar-adapted observer, we have that
\begin{equation}
\K^a{}_b u^b=\mathcal A u^a+\alpha^a,
\qquad
8\pi\mathcal Q_a^{(\phi)}=-\M\alpha_a \, ,
\label{eq:principal-leaf-flux-app}
\end{equation}
which implies that
\begin{equation}
\K^a{}_b u^b\propto u^a
\quad\Longleftrightarrow\quad
\alpha_a=0
\quad\Longleftrightarrow\quad
\mathcal Q_a^{(\phi)}=0,
\qquad
(\M\neq0).
\label{eq:principal-leaf-landau-equivalence}
\end{equation}
The existence of such an observer is not guaranteed, as already noted in Sec.~\ref{subsec:causal-preferred-frames}.

For $\M=0$, Eq.~\eqref{eq:leaf-stress-mixed-form-app} reduces to
\begin{equation}
8\pi\tau_{ab}=\mathcal P_{\ell}\ell_{ab}.
\end{equation}
The intrinsic leaf stress is then isotropic, so every scalar-adapted observer is leaf-Landau, while the existence of a principal-curvature observer remains an independent property of $\K^a{}_b$. Thus, the leaf-Landau condition is automatic, but the full Landau condition is not: in addition to the identically vanishing leaf-spatial flux $\mathcal Q_a^{(\phi)}$, one must also impose $\mathcal Q^{(\phi)}=0$, equivalently $u^aj_a=0$.

\subsection{Compatibility and full scalar-adapted Landau frames}
\label{app:full-landau}

The full four-dimensional energy flux relative to a scalar-adapted observer is
\begin{equation}
q_a^{(\phi)}=\mathcal Q^{(\phi)}s_a+\mathcal Q_a^{(\phi)}.
\label{eq:full-energy-flux-app}
\end{equation}
Equation~\eqref{eq:traction-observer-decomposition-app} gives $u^aj_a=-\mathcal Q^{(\phi)}$. A scalar-adapted observer is therefore a Landau observer of the complete effective stress-energy tensor if and only if both flux components vanish. For $\M\neq0$, using Eq.~\eqref{eq:principal-leaf-landau-equivalence}, the necessary and sufficient criterion becomes
\begin{equation}
\K^a{}_b u^b\propto u^a,
\qquad
u^aj_a=0.
\label{eq:full-scalar-adapted-Landau-criterion-app}
\end{equation}
Equivalently, $\K^a{}_b$ must possess a timelike eigendirection orthogonal to $j_a$. Table~\ref{tab:landau-compatibility} gives the resulting pointwise classification.

\begin{table}[!htbp]
\begingroup
\small
\setlength{\tabcolsep}{5pt}
\renewcommand{\arraystretch}{1.45}

\caption{Pointwise compatibility of traction-adapted, principal-curvature, and full Landau frames, for $\M\neq0$.}
\label{tab:landau-compatibility}

\begin{tabularx}{\textwidth}
{@{}L{2.7cm}L{3.6cm}Y Y@{}}
\toprule
Causal type
&
Full Landau observer
&
Compatibility condition
&
Outcome
\\
\midrule

Timelike, $j^2<0$
&
Does not exist
&
The traction-rest observer is also a principal-curvature observer if
and only if $\K^a{}_b j^b\propto j^a$
&
In that compatible frame,
$
  q_a^{(\phi)}
  =
  \mathcal Q^{(\phi)}s_a
  \neq0
$
and hence it is not a full Landau frame.
\\
\addlinespace[0.5em]

Spacelike, $j^2>0$
&
Exists if and only if $\K^a{}_b$ possesses a timelike eigendirection
orthogonal to $j_a$
&
If the principal timelike eigendirection is unique, the condition
reduces to $u_{\mathrm P}^a j_a=0$
&
The resulting observer satisfies
$\mathcal Q^{(\phi)}
  =
  0
  =
  \mathcal Q_a^{(\phi)}$ and
  $\widetilde\pi_a^{(\phi)}
  =
  j_a$.
Thus the full energy flux vanishes, while the mixed stress survives.
\\
\addlinespace[0.5em]

Null, $j_a\neq0$
&
Does not exist
&
A principal-curvature observer, and hence a leaf-Landau observer, may
still exist.
&
The longitudinal flux $\mathcal Q^{(\phi)}$ cannot vanish in any
scalar-adapted frame.
\\
\addlinespace[0.5em]

$j_a=0$
&
Exists if and only if $\K^a{}_b$ possesses a timelike eigendirection
&
The traction imposes no additional restriction.
&
Every principal-curvature observer is also a full scalar-adapted
Landau observer.
\\

\bottomrule
\end{tabularx}
\endgroup
\end{table}

When $\M=0$, the leaf-spatial flux $\mathcal Q_a^{(\phi)}$ vanishes identically. A full scalar-adapted Landau observer then exists if and only if a unit timelike vector tangent to the leaf can be chosen orthogonal to $j_a$, namely when $j^2>0$ or $j_a=0$. The static spherically symmetric realization of this classification, in which the static observer is simultaneously a principal-curvature, leaf-Landau, and full Landau, is discussed in Sec.~\ref{subsec:static-scalar-layer}.

\bigskip
\bigskip

\bibliographystyle{utphys}
\bibliography{references}

\end{document}